\documentclass[preprint,10pt,final]{elsarticle}
\usepackage{etex}
\usepackage{amssymb,amsbsy,exscale,amsmath,amsfonts,amscd,amsthm}
\usepackage[letterpaper]{geometry}
\usepackage[usenames,dvipsnames]{color}
\usepackage{bm}
\usepackage{soul}
\usepackage[ampersand]{easylist}
\usepackage{graphicx}
\graphicspath{{./images/}}
\usepackage{parskip}
\usepackage{longtable}
\usepackage{subcaption}
\usepackage[usenames,dvipsnames]{color}
\usepackage{multirow}
\usepackage{algorithmicx}
\usepackage[ruled]{algorithm}
\usepackage{algpseudocode}
\biboptions{sort&compress}
\usepackage{pifont}
\usepackage{tikz}
\usepackage{epstopdf}
\usepackage[nottoc]{tocbibind}
\usetikzlibrary{mindmap,backgrounds}
\usepackage{stmaryrd}
\usepackage{mathrsfs}
\usepackage[format=hang,font=footnotesize]{caption}
\usepackage[format=hang,font=footnotesize]{subcaption}
\usepackage{epigraph}
\usepackage{array}
\usepackage{booktabs,makecell}
\usepackage{parskip}
\usepackage{mathtools}
\usepackage[mathscr]{euscript}
\usepackage{upgreek}
\usepackage[overload]{empheq}
\usepackage{xspace}
\usepackage{arydshln}
\usepackage{cmap}
\usepackage[T1]{fontenc}
\usepackage{url}

\newcommand{\lb}{\left[}
\newcommand{\rb}{\right]}
\newcommand{\lp}{\left(}
\newcommand{\rp}{\right)}

\DeclareCaptionLabelFormat{andtable}{#1~#2 \& \tablename~\thetable}

\newcommand{\mbb}[1]{\mathbb{#1}}

\newcommand{\mcal}[1]{\mathcal{#1}}

\newcommand{\abs}[1]{\left\lvert{#1}\right\rvert}
\providecommand{\norm}[1]{\lVert#1\rVert}

\newcommand{\curl}{\textrm{curl}}
\newcommand{\dive}{\textrm{div}}

\newcommand{\mcT}[1]{\mathcal{T}_h^{#1}}
\newcommand{\bx}{{\bf x}}

\newcommand{\bn}{{\bf n}}

\newcommand{\bt}{{\bf t}}
\newcommand{\bN}{{\bf N}}

\newcommand{\bE}{{\bf E}}
\newcommand{\bH}{{\bf H}}
\newcommand{\bV}{{\bf V}}
\newcommand{\bJ}{{\bf J}}

\newcommand{\bS}{{\bf S}}

\newcommand{\bv}{{\bf v}}

\newcommand{\bu}{{\bf u}}

\newcommand{\bQ}{{\bf Q}}
\newcommand{\mm}{$\upmu$m\xspace}

\newcommand{\pow}{\varsigma}
\newcommand{\fie}{\pi}
\newcommand{\ext}{\sigma_{ext}}

\makeatletter
\newcommand*\rel@kern[1]{\kern#1\dimexpr\macc@kerna}
\newcommand*\widebar[1]{%
  \begingroup
  \def\mathaccent##1##2{%
    \rel@kern{0.8}%
    \overline{\rel@kern{-0.8}\macc@nucleus\rel@kern{0.2}}%
    \rel@kern{-0.2}%
  }%
  \macc@depth\@ne
  \let\math@bgroup\@empty \let\math@egroup\macc@set@skewchar
  \mathsurround\z@ \frozen@everymath{\mathgroup\macc@group\relax}%
  \macc@set@skewchar\relax
  \let\mathaccentV\macc@nested@a
  \macc@nested@a\relax111{#1}%
  \endgroup
}
\makeatother

\newcommand{\COMSOL}{\textsc{Comsol }{Multiphysics}\xspace}
\newtheorem{theorem}{Proposition}

 \def\etal{{\it et al.}\xspace}

\usepackage[most]{tcolorbox}
\newtcolorbox{rev1}{colback=pink,colframe=white,coltext=black,breakable}
\newtcolorbox{rev2}{colback=lightgray,colframe=white,coltext=black,breakable}

\journal{Journal of Computational Physics}

\begin{document}

\begin{frontmatter}


\title{A hybridizable discontinuous Galerkin method for computing nonlocal electromagnetic effects in three-dimensional metallic nanostructures}

\author[mit]{F.~Vidal-Codina\corref{cor1}}
\ead{fvidal@mit.edu}
\author[mit]{N.~C.~Nguyen}
\ead{cuongng@mit.edu}
\author[umn]{S.-H.~Oh}
\ead{sang@umn.edu}
\author[mit]{J.~Peraire}
\ead{peraire@mit.edu}

\cortext[cor1]{Corresponding author}
 \address[mit]{Department of Aeronautics and Astronautics, Massachusetts Institute of Technology, Cambridge, MA 02139, USA}
\address[umn]{Department of Electrical and Computer Engineering, University of Minnesota, Minneapolis, MN 55455, USA}

 \begin{abstract}
The interaction of light with metallic nanostructures produces a collective excitation of electrons at the metal surface, also known as surface plasmons. These collective excitations lead to resonances that enable the confinement of light in deep-subwavelength regions, thereby leading to large near-field enhancements. The simulation of plasmon resonances presents notable challenges. From the modeling perspective, the realistic behavior of conduction-band electrons in metallic nanostructures is not captured by Maxwell's equations, thus requiring additional modeling. From the simulation perspective, the disparity in length scales stemming from the extreme field localization demands efficient and  accurate numerical methods.

In this paper, we develop the hybridizable discontinuous Galerkin (HDG) method to solve Maxwell's equations augmented with the hydrodynamic model for the conduction-band electrons in noble metals. This method enables the efficient simulation of plasmonic nanostructures while accounting for the nonlocal interactions between electrons and the incident light. We introduce a novel postprocessing scheme to recover superconvergent solutions and demonstrate the convergence of the proposed HDG method for the simulation of a 2D gold nanowire and a 3D periodic annular nanogap structure. The results of the hydrodynamic model are compared to those of a simplified local response model, showing that differences between them can be significant at the nanoscale.

 \end{abstract}
  
\end{frontmatter}

\section{Introduction}

The field of plasmonics \cite{maier2007plasmonics,ozbay2006plasmonics} studies the collective excitation of conduction-band electrons in metallic nanostructures. These excitations, or plasmon resonances, enable the confinement of light in lengths several orders of magnitude smaller than the wavelength of light, leading to enormous near-field enhancements of the incident wave. The excitation of plasmons is magnified near the corners or sharp features of metallic nanoparticles, or within gaps formed by metallic structures at the nanoscale. Moreover, the extreme confinement and enhancement properties provide unparalleled means for the manipulation of light and its interaction with metals, at scales well beyond the diffraction limit. As a result, the field of plasmonics has motivated applications for sensing \cite{vspavckova2016optical}, energy harvesting \cite{brongersma2016plasmonic},  near-field scanning microscopy \cite{novotny2011antennas}, plasmonic waveguiding and lasing \cite{smalley2016amplification}.

Plasmonic phenomena are governed by the propagation of electromagnetic waves. These waves propagate through dielectric as well as metallic media, and several models have been proposed to characterize the behavior of metals. The most common approach to simulate plasmonic structures is to solve Maxwell's equations in both the metal and the dielectric, and account for the losses in the metal through a complex permittivity in the metal given by Drude's model \cite{drude1900elektronentheorie}. The effect of the complex permittivity in the metal is to quickly dampen the electromagnetic wave away from the interface. This approach assumes the electrons in the valence band are fully detached from the ions, thus only accounting for electron-electron and electron-ion collisions. The Drude model has limitations due to simplifications in the description of the electron motion that appear at nanometer scales, where nonlocal interaction effects between electrons become predominant \cite{fitzgerald2016quantum,zhu2016quantum,garcia2008nonlocal}. To account for these long-range interactions, the mathematical model must be enhanced. In this work, we consider the hydrodynamic model (HM) for noble metals, first introduced in the 1970s \cite{eguiluz1975influence}, which models the inter-electron coupling by including a hydrodynamic pressure term. The resulting model is solved simultaneously with Maxwell's equations. For noble metal structures with nanometric and subnanometric features, the HM predicts lower field enhancements and resonance blue-shifts, which are in better agreement with experimental data than the results computed with the Drude model \cite{toscano2012modified,raza2015review}. 

The ability to accurately model and simulate electromagnetic wave propagation problems for plasmonic applications requires capabilities that challenge traditional simulation techniques. The problems of interest involve the interaction of long-wavelength electromagnetic waves (\mm and mm) with nanometric cavities for potential applications in sensing and spectroscopy. Additionally, plasmonic phenomena are characterized by the extreme confinement and tight localization of fields in nanometer-wide apertures, nanoparticles, nanometric sharp tips, and even atomically thick materials. As a consequence, the discretizations required to attain accurate simulations need to be adaptive (to concentrate the degrees of freedom in the regions of interest) and anisotropic (to properly capture boundary-layer type structures that appear at the interface of metallic nanostructures).

The first and most widely used method for computational electromagnetics is the finite-difference time-domain (FDTD) algorithm \cite{taflove2005computational,kunz1993finite}, which discretizes both space and time using Yee's scheme \cite{yee1966numerical}. The main advantage of Yee's scheme is its simplicity and efficiency, due to the use of staggered Cartesian grids and second-order schemes for both space and time. The main limitation of FDTD is their extension to complex geometries with complex features, since Cartesian grids can only approximate these irregular boundaries in a stair-cased manner. The FDTD method has recently been applied to the hydrodynamic model for the simulation of 2D nanoparticles \cite{mcmahon2010calculating}.

Finite-volume time-domain (FVTD) methods have also been devised to solve Maxwell's equations, leveraging high-order Godunov schemes to deal with the hyperbolicity of the system \cite{munz2000divergence,ismagilov2015second}. The use of high-order Godunov schemes on a single control volume is appealing, as it renders methods that are amenable to mesh refinement and adaptation, in addition to being low dissipative and dispersive. More recently, there has been an effort to fuse these high-order Godunov schemes from FVTD with the staggering techniques from FDTD, resulting in a new generation of FVTD methods \cite{balsar2017computational,balsar2018computational} that are constraint-preserving, high-order accurate, A-stable, and that accommodate significant variations of material properties at media interfaces.

Finite element (FE) methods \cite{jin2015finite} are popular techniques for wave propagation problems, thanks to their ability to handle heterogeneous media and complex geometries with the use of unstructured grids. The class of face/edge elements introduced by N\'{e}d\'{e}lec \cite{nedelec1980mixed} have been extensively used to simulate electromagnetic wave propagation, and have been shown to avoid the problem of spurious modes \cite{bossavit1990solving} by appropriately choosing the approximation spaces. 
A commonly used implementation of edge elements for Maxwell's equations is the one provided by the RF Module of \COMSOL \cite{comsol}, which has been extended to include the hydrodynamic model \cite{toscano2012modified,ciraci2012probing}. Additionally, a frequency-domain implementation of the hydrodynamic model based on edge elements has been applied to the numerical simulation of 2D grooves and nanowires \cite{hiremath2012numerical}. 

An attractive alternative to edge elements is the class of discontinuous Galerkin (DG) methods \cite{ bassi1997high,cockburn1998local}. These methods approximate each component of the vector solution independently using standard finite element spaces within each discretization element. The solution across elements is discontinuous, and continuity of the flux is enforced weakly across element interfaces. The DG method with explicit time integration was applied to solve the time-domain Maxwell's equations \cite{hesthaven2002nodal}, and has been further developed to simulate wave propagation phenomena through metamaterials at the nanoscale \cite{busch2011discontinuous}, as well as for dispersive media \cite{lu2004discontinuous,ji2007high,lanteri2013convergence} and more recently for 2D dimers using the hydrodynamic model \cite{schmitt2016dgtd}. DG methods face disadvantages when used for practical 3D applications in the frequency domain or in the time domain with implicit time integration, due to the computational burden that arises from nodal duplication at the interfaces. This shortcoming motivated the development of the hybridizable discontinuous Galerkin (HDG) method, first introduced in \cite{cockburn2009unified} for elliptic problems, subsequently analyzed in \cite{cockburn2008superconvergent,cockburn2009superconvergent}, and later extended to a wide variety of partial differential equations (PDEs) \cite{nguyen2009linearCD,nguyen2011ns}. More specifically, the HDG has proven very effective for acoustics and elastodynamics \cite{nguyen2011acoustic,saa2012binary} as well as time-harmonic Maxwell's equations in two dimensions \cite{nguyen2011maxwell} and three dimensions \cite{li2013hybridizable}. An additional attractive feature of the HDG method is that, unlike other DG methods, it has optimal convergence rates for both the solution and the flux. As a consequence, its flux superconvergence properties can be exploited to devise a local postprocess that increases the convergence rate of the approximate solution by one order.

The main contribution of this paper is a high-order numerical scheme, the HDG method, to simulate the interaction of light with metallic nanostructures by solving the frequency-domain Maxwell's equations coupled with the hydrodynamic model for the conduction-band electrons of noble metals. There are several features of the HDG method that make it particularly attractive for computational electromagnetics: (1) it can be used on general unstructured meshes, thus allowing complex geometries and facilitating the use of adaptive discretizations; (2) it is high-order accurate, meaning it exhibits low dissipation and dispersion and is therefore well suited for wave propagation problems; (3) the linear system that needs to be solved comprises only a reduced number of degrees of freedom, defined on the faces of the discretization cells; (4) the treatment of boundary conditions is naturally incorporated in the weak formulation; (5) it does not require special approximation spaces such as curl-conforming subspaces; and (6) it can easily accommodate material contrasts at the interfaces of several orders of magnitude.

This article is organized as follows. In Section \ref{sec:problem}, we introduce the equations and notation used throughout the paper. In Section \ref{sec:hdg}, we introduce the HDG method to solve the hydrodynamic model for metals in frequency domain, and discuss the implementation  and postprocessing strategies. In Section \ref{sec:res}, we present numerical results to assess the performance of the HDG method. We finalize in Section \ref{sec:conc} by providing some concluding remarks.

\section{Modeling optical response in metallic nanostructures}\label{sec:problem}

\subsection{Maxwell's equations in a metal}
The electric $\widebar{\mcal{E}}(\widebar{\bx},\widebar{t})$ and magnetic $\widebar{\mcal{H}}(\widebar{\bx},\widebar{t})$ fields, along with the electric displacement $\widebar{\mcal{D}}$ and magnetic flux density $\widebar{\mcal{B}}$, satisfy Maxwell's equations in a metallic domain $\overline{\Omega}$
\begin{equation}\label{eq:maxwell_s}
 \begin{aligned}
\nabla \times \widebar{\mcal{E}} + \partial_{\widebar{t}}\widebar{\mcal{B}}&= 0 \quad \mbox{(Amp\`{e}re's law)}, \\
\nabla \times \widebar{\mcal{H}} -\partial_{\widebar{t}} \widebar{\mcal{D}}&=  \widebar{\mcal{J}}_{\textnormal{im}} \quad \mbox{(Faraday's law)},\\
\nabla\cdot \widebar{\mcal{D}}&= \widebar{\rho}_{\textnormal{im}},\quad \mbox{(Gauss's law)},\\
\nabla \cdot \widebar{\mcal{B}}& = 0, \quad \mbox{(magnetic Gauss's law)},
\end{aligned}
\end{equation}
where $\widebar{\mcal{J}}_{\textnormal{im}} $ represents the impressed electric current and $\widebar{\rho}_{\textnormal{im}} $ the impressed volume charge density. In addition, we have the following constitutive relations
\begin{equation}\label{eq:constitutive_s}
 \begin{aligned}
\widebar{\mcal{B}} &= \widebar{\mu} \widebar{\mcal{H}}\,,\\
\widebar{\mcal{D}} &= \varepsilon_0\widebar{\mcal{E}} + \widebar{\mcal{P}}+ \widebar{\mcal{P}}_\infty = \widebar{\varepsilon}_\infty\widebar{\mcal{E}} + \widebar{\mcal{P}}\,, \\
\nabla \cdot \widebar{\mcal{P}} &= -\widebar{\rho}\,\\
\widebar{\mcal{J}} &= \partial_{\widebar{t}}\widebar{\mcal{P}}\,.
\end{aligned}
\end{equation}
The polarization density $\widebar{\mcal{P}}$ represents the density of permanent or induced electric dipole moments due to free electrons. Conversely, the background polarization $\widebar{\mcal{P}}_\infty = (\widebar{\varepsilon}_\infty-\varepsilon_0)\widebar{\mcal{E}}$ represents the polarization of the bound electrons in the valence band. The last two relations relate the polarization density $\widebar{\mcal{P}}$ to the internal current $\widebar{\mcal{J}}$ and internal charge density $\widebar{\rho}$. The total charge density $\widebar{\rho}_{\textnormal{tot}}$ and total electric current $\widebar{\mcal{J}}_{\textnormal{tot}}$ are the summation of both the impressed and the internal contributions. In this paper, we assume there are no impressed currents and charges, hence $\widebar{\rho}_{\textnormal{tot}} =\widebar{\rho} $ and $\widebar{\mcal{J}}_{\textnormal{tot}} = \widebar{\mcal{J}}$.

\subsection{Hydrodynamic model}
A hydrodynamic model for the free electron gas was introduced in the 1970s \cite{eguiluz1975influence}. This model, despite neglecting quantum phenomena such as quantum tunneling and quantum oscillations, introduces a hydrodynamic pressure term that accounts for the nonlocal coupling of the conduction-band electrons that becomes relevant in nanometric regimes. Hence, it is referred to as nonlocal model or hydrodynamic model (HM) for noble metals.

We provide a brief derivation of the hydrodynamic model, and point the reader to \cite{eguiluz1975influence,boardman1982electromagnetic,pitarke2006theory} for a more thorough derivation. We introduce the electron density $n(\widebar{\bx},\widebar{t})$, the electron pressure $p(\widebar{\bx},\widebar{t})$ and the hydrodynamic velocity ${\bf v}(\widebar{\bx},\widebar{t})$, which are related by the continuity equation as $\partial_{\widebar{t}} n= -\nabla\cdot (n\bv)$. The equation of motion for the electron fluid under a macroscopic electromagnetic field is described as
\begin{equation}\label{eq:nonlocal0}
 m_e(\partial_{\widebar{t}} +\bv\cdot \nabla + \widebar{\gamma})\bv = -e(\widebar{\mcal{E}} + \bv\times\widebar{\mcal{H}}) - \frac{\nabla p}{n}\;,
\end{equation}
where $m_e$ is the effective electron mass, $e$ is the electron charge and $\widebar{\gamma}$ is a damping constant related to the collision rate of the electrons. In order to simplify the above equation, we linearize the electron density field around the constant equilibrium density of the electron gas $n_0$, such that $n(\widebar{\bx},\widebar{t}) \approx n_0 + n_1(\widebar{\bx},\widebar{t})$; neglect the high order term for the derivative of the hydrodynamic velocity $\bv\cdot\nabla\bv$; and also neglect the effect of the magnetic field, since the electron fluid is driven mainly by the electric field. In addition, we simplify the pressure term in \eqref{eq:nonlocal0} assuming a Thomas-Fermi model where only the kinetic energy is relevant, that is 
\begin{equation*}
 \frac{\nabla p}{n} \approx m_e \widebar{\beta}^2 \frac{\nabla n_1}{n_0}\;.
\end{equation*}
The quantum parameter $\widebar{\beta}$, which represents the nonlocality, is usually expressed \cite{boardman1982electromagnetic} in terms of the Fermi velocity $\widebar{v}_F$ as $\widebar{\beta} = \sqrt{3/5}\widebar{v}_F$. Using the assumptions above, the equation of motion for the electron fluid can be simplified as
\begin{equation*}
 m_e(\partial_{\widebar{t}} +\widebar{\gamma})\bv = -e\widebar{\mcal{E}} - m_e\widebar{\beta}^2 \frac{\nabla n_1}{n_0}\;,
\end{equation*}
and if we differentiate with respect to time, we arrive at 
\begin{equation}\label{eq:nonlocal1}
 m_e(\partial_{\widebar{t}\widebar{t}} +\widebar{\gamma}\partial_{\widebar{t}})\bv = -e\partial_{\widebar{t}}\widebar{\mcal{E}} + m_e\widebar{\beta}^2 \nabla(\nabla\cdot \bv)\;,
\end{equation}
where the last term is obtained by linearizing the continuity equation $\partial_{\widebar{t}} n_1= -\nabla\cdot (n_0\bv)$ and neglecting the high-order term $\nabla\cdot(n_1\bv)$. Using the relation between the electric current and the electron gas density $\widebar{\mcal{J}} = -en\bv$, and multiplying \eqref{eq:nonlocal1} by $-en/m_e$, we obtain
\begin{equation}\label{eq:hm_time}
\partial_{\widebar{t}\widebar{t}}\widebar{\mcal{J}} +\widebar{\gamma}\partial_{\widebar{t}}\widebar{\mcal{J}} = \frac{e^2n}{m_e\varepsilon_0}\varepsilon_0\partial_{\widebar{t}}\widebar{\mcal{E}}+ \widebar{\beta}^2 \nabla(\nabla\cdot \widebar{\mcal{J}})\;.
\end{equation}
This equation, which prescribes a nonlocal relationship between the electric field and the electric current, needs to be solved simultaneously with Maxwell's equations \eqref{eq:maxwell_s} inside the metal. The parameter involved in the third term is the square of the metal's plasma frequency $\widebar{\omega}_p$, defined as $\widebar{\omega}_p = e\sqrt{n/(m_e\varepsilon_0)}$. The plasma frequency represents the frequency above which the conduction electrons are not able to oscillate in phase with the incident light, thus effectively impeding the cancellation the incoming wave. That is, for frequencies larger than the plasma frequency the incident wave is allowed to propagate through the metal, although with losses.

It is convenient to nondimensionalize Maxwell's equations. We use the following scalings for the electromagnetic fields
 \begin{align*}
\widebar{\bx} &= \bx /L_c,\quad \widebar{t} = t c_0/L_c,\quad \widebar{\mcal{E}} = \alpha Z_0 \mcal{E}, \quad  \widebar{\mcal{H}} = \alpha \mcal{H},\\
\widebar{\mcal{D}} =& \varepsilon_0 \alpha Z_0 \mcal{D},\quad \widebar{\mcal{B}} = \mu_0 \alpha Z_0 \mcal{B},\quad \widebar{\mcal{J}} =  \alpha \mcal{J}/L_c, 
\end{align*}
where $L_c$ is a reference length scale, $\alpha$ is a reference magnetic field and $\varepsilon_0,\,\mu_0$ are the free-space permittivity and permeability, that relate to the free-space speed of light $c_0 = 1/\sqrt{\varepsilon_0\mu_0}$ and free-space impedance $Z_0 = \sqrt{\mu_0/\varepsilon_0}$. For a non-magnetic medium $(\widebar{\mu} = \mu_0)$, applying the scalings above to Maxwell's equations \eqref{eq:maxwell_s}, the constitutive relations \eqref{eq:constitutive_s} and the hydrodynamic pressure equation \eqref{eq:hm_time}, we obtain 
 \begin{align*}
\nabla \times \mcal{E} + \partial_t\mcal{H}&= 0 ,\\
\nabla \times \mcal{H} -\partial_t \varepsilon_\infty\mcal{E}&= \mcal{J} , \\
 {\beta}^2 \nabla(\nabla\cdot {\mcal{J}}) - \partial_{{t}{t}}{\mcal{J}} -{\gamma}\partial_{{t}}{\mcal{J}} &=- \omega_p^2\partial_{{t}}{\mcal{E}},\\
\nabla\cdot \lp\varepsilon_\infty \mcal{E}\rp&= \rho ,\\
\nabla \cdot \mcal{H}& = 0 ,
\end{align*}
with the nondimensional variables $\varepsilon_\infty = \widebar{\varepsilon}_\infty/\varepsilon_0$, $\omega_p =\widebar{\omega}_pL_c/c_0$, $\gamma=\widebar{\gamma}L_c/c_0$ and $\beta = \widebar{\beta}/c_0$. 

Using the linearity of Maxwell's equations we can write, for a given angular frequency $\omega$, the components of, for instance, the electric field as $\mcal{E}(\bx,t) = \Re\lbrace \bE(\bx) \exp(-i\omega t)\rbrace$. Consequently, the time-domain equations are recast into the frequency domain through the transformation $\partial_t \mapsto -i\omega$. Hence, the frequency-domain Maxwell's equations with the hydrodynamic model for metals are given by
\begin{equation}\label{eq:maxwellhydro}
 \begin{aligned}
\nabla \times \bE -i\omega\bH &= 0 ,\\
\nabla \times \bH + i\omega\varepsilon_\infty \bE&= \bJ , \\
\beta^2 \nabla(\nabla\cdot \bJ) + \omega(\omega + i\gamma)\bJ &= i\omega\omega_p^2\bE,\\
\nabla\cdot \lp \varepsilon_\infty \bE \rp&= \uprho  ,\\
\nabla \cdot \bH & = 0  .
\end{aligned}
\end{equation}
The system above is completed with boundary conditions
\begin{align*}
\label{eq:maxwellhydroV_BC}
&\bn\times\bE\times\bn  = \bE_\partial, \quad  \mbox{on }\partial\overline{\Omega}_E,\\
&\bn\times\bH = \bn\times\bH_\partial,\quad  \mbox{on }\partial\overline{\Omega}_V,\\
&\bn\cdot\bJ  = 0, \quad \mbox{on }\partial\overline{\Omega}.
\end{align*}
where $\partial \overline{\Omega} = \partial\overline{\Omega}_E \cup \partial\overline{\Omega}_V $. The last boundary condition \cite{boardman1976surface} prescribes a vanishing normal electric current at the interface. Physically, it simulates a no electron spill-out condition, that is the electrons are precluded from leaving the metal. Effects such as electron tunneling, a quantum phenomenon that becomes relevant in subnanometric regimes, are therefore not included in the HM.

The more simplistic Drude model, also known local response approximation (LRA), may be obtained from \eqref{eq:maxwellhydro} by setting $\beta = 0$, which recovers a local relation between the electric field and internal electric current $\bJ = \frac{i\omega\omega_p^2}{\omega(\omega+i\gamma)}\bE$ (Ohm's law). The complex Drude permittivity can therefore be written as $\varepsilon(\omega) = \varepsilon_\infty - \omega_p^2/(\omega(\omega + i\gamma))$. The Drude model for metals is attractive for its simplicity, and produces acceptable results for many electromagnetic applications. Nonetheless, the assumption that all electrons exhibit a local behavior produces unphysical results for frequencies close to the plasma frequency \cite{ruppin2001extinction} and for geometries and features below ten nanometers \cite{romero2006plasmons,zhu2016quantum}. In these regimes the HM is able to capture more accurate electromagnetic responses than the LRA.

Another difference between the LRA and the HM is the distribution of the internal charge density $\uprho$, defined as $i\omega\uprho = \nabla\cdot \bJ$. The solutions provided by the local model infinitely squash $\uprho$ at the metal surface, which results in a Dirac delta at the metal-dielectric interface. That is, the metal acts as a hard wall for the incoming EM wave, and impedes propagation through it. Conversely, the electron pressure term in the hydrodynamic model regularizes the induced charge density by smoothing its profile, thus allowing the penetration of the incident field. The spreading distance experienced by the charge density is on the order of the length $\updelta= \beta/\omega_p$, introduced in \cite{ciraci2013hydrodynamic}. 


\section{HDG method for the hydrodynamic model}\label{sec:hdg}

\subsection{Approximation spaces}

We first review the basic notation, operators and approximation spaces needed for the HDG method for Maxwell's equations in 3D, following \cite{nguyen2011maxwell}. We denote by $\mcT{}$ a triangulation of disjoint regular elements $T$ that partition an open domain $\mcal{D}\in\mbb{R}^3$. The set of element boundaries is then defined as $\partial \mcT{}:=\lbrace \partial T:\,T\in\mcT{} \rbrace$. For an arbitrary element $T\in\mcT{}$, $F = \partial T\cap \partial \mcal{D}$ is a boundary face if it has a nonzero 2D Lebesgue measure. Any pair of elements $T^+$ and $T^-$ share an interior face $F = \partial T^+ \cap \partial T^-$ if its 2D Lebesgue measure is nonzero. We finally denote by $\mcal{E}_h^o$ and $\mcal{E}_h^\partial$ the set of interior and boundary faces respectively, and the total set of faces $\mcal{E}_h = \mcal{E}_h^o\cup \mcal{E}_h^\partial$.

Let $\bn^+$ and $\bn^-$ be the outward-pointing unit normal vectors on the neighboring elements $T^+,\,T^-$, respectively. We further use $\bu^\pm$ to denote the trace of $\bu$ on $F$ from the interior of $T^\pm$. The jump $\llbracket\cdot\rrbracket$ for an interior face $F\in\mcal{E}_h^o$ is defined as
\begin{equation*}
 \llbracket \bu \odot \bn \rrbracket = \bu^+\odot\bn^+ + \bu^-\odot\bn^-,
\end{equation*}
and is single valued for a boundary face $F\in\mcal{E}_h^\partial$ with outward normal $\bn$, that is
\begin{equation*}
 \llbracket \bu \odot \bn \rrbracket = \bu\odot\bn,
\end{equation*}
where the binary operation $\odot$ refers to either $\cdot$ or $\times$. The tangential $\bu^t$ and normal $\bu^n$ components of $\bu$, for which $\bu =\bu^t + \bu^n$, are then represented as
\begin{equation*}
 \bu^t:= \bn \times(\bu\times\bn),\qquad \bu^n := \bn(\bu\cdot\bn).
\end{equation*}

Let $\bm L^2(\mcal{D})\equiv [L^2(\mcal{D})]^3$ denote the Lebesgue space of square integrable functions with three components and $H^1(\mcal{D})$ the Hilbert space with $H^1(\mcal{D}) = \lbrace v\in L^2(\mcal{D}):\,\int_\mcal{D}\abs{\nabla v}^2 <\infty\rbrace$. We introduce the curl-conforming space
\begin{equation*}
 \bm H^{\curl}(\mcal{D}) = \lbrace \bu \in \bm L^2(\mcal{D}):\nabla\times \bu \in \bm L^2(\mcal{D}) \rbrace
\end{equation*}
with associated norm $\norm{\bu}^2_{ \bm H^{\curl}(\mcal{D})} = \int_{\mcal{D}} \abs{\bu}^2 + \abs{\nabla\times\bu}^2$, as well as the div-conforming space
\begin{equation*}
 \bm H^{\dive}(\mcal{D}) = \lbrace \bu \in \bm L^2(\mcal{D}):\nabla\cdot \bu \in L^2(\mcal{D}) \rbrace
\end{equation*}
with associated norm $\norm{\bu}^2_{ \bm H^{\dive}(\mcal{D})} = \int_\mcal{D} \abs{\bu}^2 + \abs{\nabla\cdot\bu}^2$. 

Let $\mcal{P}^p(\mcal{D})$ denote the space of complex-valued polynomials of degree at most $p$ on $\mcal{D}$. We introduce the following approximation spaces 
\begin{align*}
W_h &= \{w\in L^2(\mcal{D}) : w|_T \in \mcal{P}^{p}(T), \;\forall T\in \mathcal{T}_h\},\\
\bm W_h &= \{\bm \xi\in \bm L^2(\mcal{D}) : \bm \xi|_T \in \lb \mcal{P}^{p}(T) \rb^3,\; \forall T\in \mathcal{T}_h\},\\
M_h &= \{\mu\in  L^2(\mcal{E}_h)\,: \mu|_{F} \in \mcal{P}^{p}(F),\;\forall F\in \mcal{E}_h \},\\
\bm M_h &= \{\bm\mu\in \bm L^2(\mcal{E}_h)\,: \bm\mu|_{F} \in \mcal{P}^{p}(F)\bt_1 \oplus \mcal{P}^{p}(F)\bt_2,\;\forall F\in \mcal{E}_h \},
\end{align*}
where $\bt_1,\, \bt_2$ are linearly independent vectors tangent to the face, thus naturally including the $\bm H^{\curl}$ nature of the solutions, since by construction $\bm\mu \in \bm M_h$ satisfies $\bm\mu = \bn\times(\bm\mu\times\bn) = \mu_1 \bt_1 + \mu_2 \bt_2$. The tangent vectors can be defined in terms of $\bn = (n_1,n_2,n_3$) as ${\bf t}_1 =(-n_2/n_1,1,0)$ and ${\bf t}_2 =(-n_3/n_1,0,1)$. This definition assumes that $|n_1| \ge \max(|n_2|,|n_3|)$ but analogous expressions can be obtained when $|n_2| \ge \max(|n_1|,|n_3|)$ or $|n_3| \ge \max(|n_1|,|n_2|)$ to avoid division by a small number. Boundary conditions are included by setting $\bm M_h ({\bf u}_\partial) = \lbrace \bm \mu \in \bm M_h:\, \bn\times\bm\mu = \Pi {\bf u}_\partial\; \mbox{on } \partial \mcal{D} \rbrace$ and $ M_h ({ u}_\partial) = \lbrace  \mu \in  M_h:\, \mu = \Pi { u}_\partial\; \mbox{on } \partial \mcal{D} \rbrace$, where $\Pi {\bf u}_\partial$ (respectively, $\Pi u_\partial$) is the projection of ${\bf u}_\partial$ onto $\bm M_h$ (respectively, $u_\partial$ onto $M_h$).

Finally, we define the various Hermitian products for the above finite element spaces.  The volume inner products are defined as
\begin{equation*}
 (\eta,\zeta)_{\mcT{}} := \sum_{T\in\mcT{}}(\eta,\zeta)_T,\qquad  (\bm\eta,\bm\zeta)_{\mcT{}} := \sum_{i = 1}^3(\eta_i,\zeta_i)_{\mcT{}},
\end{equation*}
and the surface inner products by
\begin{equation*}
 \langle\eta,\zeta\rangle_{\partial\mcT{}} := \sum_{T\in\mcT{}} \langle\eta,\zeta \rangle_{\partial T},\qquad   \langle\bm\eta,\bm\zeta \rangle_{\partial\mcT{}} := \sum_{i = 1}^3 \langle \bm\eta_i,\bm\zeta_i \rangle_{\partial\mcT{}}.
\end{equation*}
For two arbitrary scalar functions $\eta$ and $\zeta$, its scalar product $(\eta,\zeta)_\mcal{D}$ is the integral of $\eta\zeta^*$ on $\mcal{D}$. 

\subsection{Numerical approximation}

We now describe an HDG method to numerically solve Maxwell's equations with the hydrodynamic model \eqref{eq:maxwellhydro} for a metallic computational domain $\overline{\Omega}$, which will serve as a building block towards more complicated scenarios. We introduce additional variables $\bV = i\omega\bH$, $U = \nabla\cdot\bJ$ and rewrite system \eqref{eq:maxwellhydro} as a first order system of equations in $\overline{\Omega}$:
 \begin{empheq}
[left={\overline{\mcal{L}} \; \empheqlbrace}]{align}
\nabla \times \bE - \bV &= 0 , \notag \\
\beta^2 \nabla U + \omega(\omega + i\gamma)\bJ - i\omega\omega_p^2\bE &= 0, \label{eq:firstordersystemhydro} \\
\nabla \times \bV - \omega^2\varepsilon_\infty \bE -i\omega\bJ & =0 , \notag \\
U -\nabla\cdot \bJ &= 0 .\notag
\end{empheq}
The additional variable $U$ is related to the induced free charge density in the metal as $i\omega\uprho =  U$. 


We seek  $(\bV_h,\bE_h,\bJ_h,U_h,\widehat{\bE}_h,\widehat{U}_h) \in \bm W_h \times \bm W_h \times \bm W_h \times W_h \times \bm M_h  \times M_h$ such that
\begin{equation}\label{eq:hdg_hydro1}
 \begin{aligned}
(\bV_h,\bm \kappa)_{\mcal{T}_h} - (\bE_h,\nabla\times \bm \kappa)_{\mcal{T}_h} - \langle \widehat{\bE}_h,\bm \kappa\times\bn \rangle_{\partial\mcal{T}_h} & =0 , \\
-\beta^2 (U_h,\nabla\cdot\bm\eta)_{\mcal{T}_h} + \beta^2\langle \widehat{U}_h,\bm\eta\cdot\bn \rangle_{\partial\mcal{T}_h} + \omega(\omega + i\gamma)(\bJ_h,\bm\eta)_{\mcal{T}_h} - i\omega\omega_p^2(\bE_h,\bm\eta)_{\mcal{T}_h} &= 0,\\
(\bV_h,\nabla \times \bm \xi)_{\mcal{T}_h} + \langle \widehat{\bV}_h,\bm \xi\times\bn \rangle_{\partial\mcal{T}_h} - \omega^2(\varepsilon_\infty\bE_h,\bm \xi)_{\mcal{T}_h} - i\omega(\bJ_h,\bm \xi)_{\mcal{T}_h}  & =0 ,\\
(U_h,\zeta)_{\mcal{T}_h} - \langle \widehat{\bJ}_h\cdot \bn,\zeta \rangle_{\partial\mcal{T}_h} + (\bJ_h,\nabla \zeta)_{\mcal{T}_h} &= 0, \\
-\langle \bn\times \widehat{\bV}_h,\bm\mu \rangle_{\partial\mcal{T}_h \backslash \partial \overline{\Omega}} +\langle \widehat{\bE}_h - \bE_\partial,\bm\mu \rangle_{\partial \overline{\Omega}_E} - \langle \bn\times \widehat{\bV}_h - \bn\times {\bV}_\partial  ,\bm\mu \rangle_{\partial\overline{\Omega}_V}&=0, \\
\langle \widehat{\bJ}_h \cdot \bn ,\theta \rangle_{\partial\mcal{T}_h} &= 0, 
\end{aligned} 
\end{equation}
holds for all $(\bm \kappa,\bm\eta,\bm \xi,\zeta,\bm\mu,\theta)\in \bm W_h \times \bm W_h \times \bm W_h \times W_h \times \bm M_h \times M_h$, where  $\widehat{\bE}_h $ approximates the tangential field of $\bE$, and $\widehat{U}_h$ approximates the trace of $U$. 
We close the system by introducing expressions for the hybrid fluxes of the magnetic field and electric current field as 
\begin{equation}
\begin{aligned}\label{eq:traces}
 \widehat{\bV}_h &= {\bV}_h + \tau_t(\bE_h - \widehat{\bE}_h)\times\bn,\\
  \widehat{\bJ}_h &= {\bJ}_h - \tau_n(U_h-\widehat{U}_h)\bn.
\end{aligned} 
\end{equation}
The parameters $\tau_t,\,\tau_n$ are the stabilization parameters, defined globally to ensure the accuracy and stability of the HDG discretization. We propose the choice $\tau_t =\sqrt{\varepsilon_\infty}\omega$ and $\tau_n = 1/\updelta = \omega_p/\beta$. This choice leads to numerically stable solutions even in the presence of tightly localized fields in the metal-dielectric interface. 

Substituting  \eqref{eq:traces} in \eqref{eq:hdg_hydro1} and integrating by parts, we write the final HDG discretization of the hydrodynamic model for metals as
\begin{equation}\label{eq:hdg_hydro}
 \begin{aligned}
(\bV_h,\bm \kappa)_{\mcal{T}_h} - (\bE_h,\nabla\times \bm \kappa)_{\mcal{T}_h} - \langle \widehat{\bE}_h,\bm \kappa\times\bn \rangle_{\partial\mcal{T}_h}  &=0 , \\
-\beta^2 (U_h,\nabla\cdot\bm\eta)_{\mcal{T}_h} + \beta^2\langle \widehat{U}_h,\bm\eta\cdot\bn \rangle_{\partial\mcal{T}_h} + \omega(\omega + i\gamma)(\bJ_h,\bm\eta)_{\mcal{T}_h} - i\omega\omega_p^2(\bE_h,\bm\eta)_{\mcal{T}_h} &= 0,\\
(\nabla \times \bV_h,\bm \xi)_{\mcal{T}_h} +\langle  \tau_t[\bE_h-\widehat{\bE}_h],\bn\times\bm \xi\times\bn \rangle_{\partial\mcal{T}_h} - \omega^2(\varepsilon_\infty\bE_h,\bm \xi)_{\mcal{T}_h} - i\omega(\bJ_h,\bm \xi)_{\mcal{T}_h}  &=0 , \\
- (\nabla\cdot {\bJ}_h,\zeta )_{\mcal{T}_h} + (U_h,\zeta)_{\mcal{T}_h}  + \langle \tau_n U_h,\zeta \rangle_{\partial\mcal{T}_h} - \langle \tau_n\widehat{U}_h,\zeta \rangle_{\partial\mcal{T}_h} &= 0,\\
-\langle \bn\times {\bV}_h + \tau_t \bE_h,\bm\mu \rangle_{\partial\mcal{T}_h \backslash \partial \overline{\Omega}} +\langle\widetilde{\tau}_t \widehat{\bE}_h,\bm\mu \rangle_{\partial\mcal{T}_h} - \langle {\bf f},\bm\mu \rangle_{\partial\overline{\Omega}}&= 0, \\
\langle {\bJ}_h \cdot \bn ,\theta \rangle_{\partial\mcal{T}_h}  -  \langle\tau_n U_h ,\theta \rangle_{\partial\mcal{T}_h} +  \langle\tau_n\widehat{U}_h ,\theta \rangle_{\partial\mcal{T}_h}&= 0.
\end{aligned} 
\end{equation}
where 
\begin{equation}\label{eq:taucompressed}
\widetilde{\tau}_t = 
\begin{cases}
 \tau_t, & \mbox{on } \partial\mcal{T}_h \backslash \partial \Omega_E \\
  1, & \mbox{on } \partial\Omega_{E} \\
\end{cases}
,
\qquad \qquad
{\bf f} = 
\begin{cases}
\bE_\partial, & \mbox{on }  \partial \Omega_E \\
-\bn\times {\bV}_\partial , & \mbox{on } \partial\Omega_{V} \\
\end{cases}
.
\end{equation}
The first four equations represent the weak formulation of equations \eqref{eq:firstordersystemhydro}, whereas the last two equations enforce zero jump in the tangential component of $\bV_h$ and in the normal component of $\bJ_h$ respectively, along with the appropriate boundary conditions. 

We now complete the definition of the HDG method for Maxwell's equation with the hydrodynamic model, by showing the method is consistent, conservative and well defined.
\newline
\begin{theorem}
 The HDG method defined by \eqref{eq:hdg_hydro} is consistent and its numerical fluxes are uniquely defined over the edges, therefore is also conservative.
\end{theorem}
\begin{proof}
 The last two equations of \eqref{eq:hdg_hydro1} imply that
 \begin{align*}
  \llbracket \bn\times\widehat{\bV}_h\rrbracket &= 0,\quad \mbox{on } \mcal{E}_h^o,\\
   \llbracket \bn\cdot\widehat{\bJ}_h\rrbracket &= 0,\quad \mbox{on } \mcal{E}_h^o.
\end{align*}
Substituting \eqref{eq:traces} into the expressions above we arrive at
 \begin{align*}
  \llbracket \bn\times{\bV}_h\rrbracket +\tau_t^+\bE_h^+ + \tau_t^- \bE_h^- - (\tau_t^+ + \tau_t^-)\widehat{\bE}_h&= 0,\quad \mbox{on } \mcal{E}_h^o,\\
  \llbracket \bn\cdot{\bJ}_h\rrbracket -\tau_t^+U_h^+ - \tau_t^-U_h^- + (\tau_t^+ + \tau_t^-)\widehat{U}_h&= 0,\quad \mbox{on } \mcal{E}_h^o.
\end{align*}
Isolating the value of the traces we get
\begin{equation}\label{eq:tracesEU}
  \begin{aligned}
  \widehat{\bE}_h &= \dfrac{\tau_t^+\bE_h^+ + \tau_t^-\bE_h^-  + \llbracket\bn\times\bV_h \rrbracket}{\tau_t^+ + \tau_t^-},\quad \mbox{on } \mcal{E}_h^o,\\
   \widehat{U}_h &= \dfrac{\tau_t^+U_h^+ + \tau_t^-U_h^- -\llbracket \bn\cdot\bJ_h\rrbracket}{\tau_t^+ + \tau_t^-} ,\quad \mbox{on } \mcal{E}_h^o,
\end{aligned}
\end{equation}
and substituting these expressions into \eqref{eq:traces} we obtain
\begin{equation}\label{eq:tracesVQ}
  \begin{aligned}
  \widehat{\bV}_h &= \dfrac{\tau_t^+\bV_h^- + \tau_t^-\bV_h^+ + \tau_t^+\tau_t^- \llbracket\bE_h\times\bn\rrbracket}{\tau_t^+ + \tau_t^-},\quad \mbox{on } \mcal{E}_h^o,\\
    \widehat{\bJ}_h &= \dfrac{\tau_t^+\bJ_h^- + \tau_t^-\bJ_h^+ - \tau_t^+\tau_t^- \llbracket U_h\bn\rrbracket}{\tau_t^+ + \tau_t^-},\quad \mbox{on } \mcal{E}_h^o.
\end{aligned}
\end{equation}
The expressions \eqref{eq:tracesEU} and \eqref{eq:tracesVQ} show that the numerical traces of the HDG method are single valued across inter-element faces, hence the HDG method is conservative by virtue of the definition of conservation introduced in \cite{arnold2002unified} for DG methods. Furthermore, since $\bE\in \bm H^{\curl}(\Omega)$ and $U\in H^1(\Omega)$, we have $\widehat{\bE} = \bE^t$ and $\widehat{U}=U$ on $\mcal{E}_h$. It follows from expressions \eqref{eq:traces} that $\widehat{\bV} = \bV$ and $\widehat{\bJ} = \bJ$. Finally, if we substitute them into the first four equations of \eqref{eq:hdg_hydro1} and integrating back by parts, we arrive at
 \begin{align*}
(\bV -\nabla\times\bE,\bm \kappa)_{\mcal{T}_h} & =0 , \\
(\beta^2 \nabla U + \omega(\omega + i\gamma)\bJ - i\omega\omega_p^2\bE ,\bm\eta)_{\mcal{T}_h} &= 0,\\
(\nabla \times \bV - \omega^2\varepsilon_\infty \bE -i\omega\bJ ,\bm \xi)_{\mcal{T}_h}  & =0 ,\\
(U -\nabla\cdot \bJ, \zeta)_{\mcal{T}_h} &= 0.
\end{align*} 
The exact solution of \eqref{eq:maxwellhydro} is therefore a solution of the HDG formulation \eqref{eq:hdg_hydro1}, thus the HDG method is consistent.
\end{proof}
In addition, it can also be shown that the solution of the HDG method proposed is unique away from the resonances.
\newline
\begin{theorem}
 Assume that both $\omega^2\varepsilon_\infty $ and $\omega(\omega+i\gamma)$ are different from the eigenvalues $\lambda_1,\,\lambda_2$ of the following eigenproblem: find $\lambda_1,\,\lambda_2\in\mbb{C}$ and $(\bN_h,\bQ_h,\bS_h,\psi_h,\widehat{\bQ}_h,\widehat{\psi}_h) \in \bm W_h \times \bm W_h \times \bm W_h \times W_h \times \bm M_h ({\bf 0}) \times M_h$ such that
 \begin{equation}\label{eq:eigenproblem}
 \begin{aligned}
(\bN_h,\bm \kappa)_{\mcal{T}_h} - (\bQ_h,\nabla\times \bm \kappa)_{\mcal{T}_h} - \langle \widehat{\bQ}_h,\bm \kappa\times\bn \rangle_{\partial\mcal{T}_h}  &=0 , \\
-\beta^2 (\psi_h,\nabla\cdot\bm\eta)_{\mcal{T}_h} + \beta^2\langle \widehat{\psi}_h,\bm\eta\cdot\bn \rangle_{\partial\mcal{T}_h}  - i\omega\omega_p^2(\bQ_h,\bm\eta)_{\mcal{T}_h} &= -\lambda_2(\bS_h,\bm\eta)_{\mcal{T}_h},\\
(\nabla \times \bN_h,\bm \xi)_{\mcal{T}_h} + \tau_t\langle \bQ_h-\widehat{\bQ}_h ,\bn\times\bm \xi\times\bn \rangle_{\partial\mcal{T}_h}  - i\omega(\bS_h,\bm \xi)_{\mcal{T}_h}  &= \lambda_1(\bQ_h,\bm \xi)_{\mcal{T}_h} , \\
- (\nabla\cdot {\bS}_h,\zeta )_{\mcal{T}_h} + (\psi_h,\zeta)_{\mcal{T}_h}  + \tau_n\langle \psi_h,\zeta \rangle_{\partial\mcal{T}_h} - \tau_n\langle \widehat{\psi}_h,\zeta \rangle_{\partial\mcal{T}_h} &= 0,\\
-\langle \bn\times {\bN}_h + \tau_t (\bQ_h-\widehat{\bQ}_h),\bm\mu \rangle_{\partial\mcal{T}_h }&= 0, \\
\langle {\bS}_h \cdot \bn ,\theta \rangle_{\partial\mcal{T}_h}  - \tau_n \langle \psi_h ,\theta \rangle_{\partial\mcal{T}_h} + \tau_n \langle\widehat{\psi}_h ,\theta \rangle_{\partial\mcal{T}_h}&= 0,
\end{aligned} 
\end{equation}
for any $(\bm \kappa,\bm\eta,\bm \xi,\zeta,\bm\mu,\theta)\in \bm W_h \times \bm W_h \times \bm W_h \times W_h \times \bm M_h({\bf 0}) \times M_h$. Furthermore, if the stabilization parameters are positive on $\partial\mcal{T}_h$, then the HDG solution $(\bV_h,\bE_h,\bJ_h,U_h,\widehat{\bE}_h,\widehat{U}_h)$ exists and is uniquely defined.
\begin{proof}
Since the square system above is linear and finite dimensional, it is sufficient to show that the trivial solution is the unique solution of \eqref{eq:hdg_hydro} if $\bE_\partial = \bV_\partial = 0$. If we take $\bm \kappa = \bV_h,\, \bm\eta  = \bJ_h,\, \bm \xi = \bE_h,\, \zeta = U_h,\, \bm \mu = \widehat{\bE}_h$ and $\theta  = \widehat{U}_h$ in \eqref{eq:hdg_hydro}, multiply the second equation by $-1/\omega_p^2$, the fourth and sixth by $\beta^2/\omega_p^2$, and add them together, we arrive at
\begin{equation*}
\begin{split}
&(\bV_h,\bV_h)_{\mcal{T}_h} + \tau_t\langle (\bE_h-\widehat{\bE}_h)\times\bn ,(\bE_h-\widehat{\bE}_h)\times\bn\rangle_{\partial\mcal{T}_h} + \dfrac{\beta^2}{\omega_p^2} (U_h,U_h)_{\mcal{T}_h} +\\& \tau_n\langle U_h-\widehat{U}_h,U_h-\widehat{U}_h\rangle_{\partial\mcal{T}_h} = \omega^2\varepsilon_\infty(\bE_h,\bE_h)_{\mcal{T}_h} + \dfrac{\omega(\omega+i\gamma)}{\omega_p^2} (\bJ_h,\bJ_h)_{\mcal{T}_h}.
\end{split}
\end{equation*}
Similarly, for the eigenproblem in \eqref{eq:eigenproblem} we have
\begin{equation*}
\begin{split}
&(\bN_h,\bN_h)_{\mcal{T}_h} + \tau_t\langle (\bQ_h-\widehat{\bQ}_h)\times\bn ,(\bQ_h-\widehat{\bQ}_h)\times\bn\rangle_{\partial\mcal{T}_h} + \dfrac{\beta^2}{\omega_p^2} (\psi_h,\psi_h)_{\mcal{T}_h} +\\& \tau_n\langle \psi_h-\widehat{\psi}_h,\psi_h-\widehat{\psi}_h\rangle_{\partial\mcal{T}_h} = \lambda_1(\bQ_h,\bQ_h)_{\mcal{T}_h} + \dfrac{\lambda_2}{\omega_p^2}(\bS_h,\bS_h)_{\mcal{T}_h}.
\end{split}
\end{equation*}
It follows from the previous two equations that both $\bE_h$ and $\bJ_h$ are zero; otherwise, $\omega^2\varepsilon_\infty $ and $\omega(\omega+i\gamma)$ must be eigenvalues of \eqref{eq:eigenproblem} which contradicts the hypothesis. As a consequence, we get
\begin{equation*}
(\bV_h,\bV_h)_{\mcal{T}_h} + \tau_t\langle \widehat{\bE}_h\times\bn ,\widehat{\bE}_h\times\bn\rangle_{\partial\mcal{T}_h} + \dfrac{\beta^2}{\omega_p^2} (U_h,U_h)_{\mcal{T}_h} + \tau_n\langle U_h-\widehat{U}_h,U_h-\widehat{U}_h\rangle_{\partial\mcal{T}_h} =0,
\end{equation*}
hence $\bV_h = 0,\,\widehat{\bE}_h=0,\,U_h = 0$ and $\widehat{U}_h = 0$ since the stabilization constants are strictly positive. In consequence, the trivial solution is the unique solution of the HDG discretization with homogeneous boundary conditions, thus completing the proof.
\end{proof}
\end{theorem}


\subsection{Implementation}\label{sec:imp}
The system of equations in \eqref{eq:hdg_hydro} is rewritten for convenience in terms of several bilinear forms. The weak formulation reads: find $(\bE_h,\bV_h,\bJ_h,U_h,\widehat{\bE}_h,\widehat{U}_h) \in \bm W_h \times \bm W_h \times \bm W_h \times W_h \times \bm M_h ({\bf 0}) \times M_h$ such that
\begin{align}
\mathscr{A}(\bV_h,\bm \kappa) - \mathscr{B}(\bE_h,\bm \kappa)-\mathscr{C}(\widehat{\bE}_h,\bm \kappa)&= 0,\notag\\
\omega(\omega+i\gamma)\mathscr{A}(\bJ_h,\bm\eta) -i\omega\omega_p^2 \mathscr{A}({\bE}_h,\bm\eta) -\beta^2 \mathscr{P}(U_h,\bm\eta) + \beta^2 \mathscr{O}(\widehat{U}_h,\bm\eta) &=0,\notag\\
\mathscr{B}(\bm \xi,\bV_h) - i\omega \mathscr{A}(\bJ_h,\bm \xi) + \mathscr{D}(\bE_h,\bm \xi) - \omega^2\mathscr{A}_\varepsilon(\bE_h,\bm \xi)-\mathscr{E}(\widehat{\bE}_h,\bm \xi) & = 0,\label{eq:hdghydro_weaksystem}\\
-\mathscr{P}(\zeta,\bJ_h) + \mathscr{H}(U_h,\zeta) - \mathscr{N}(\widehat{U}_h,\zeta) & =0,\notag\\
-\mathscr{R}(\bV_h,\bm\mu) -\mathscr{L}(\bE_h,\bm\mu) + \mathscr{M}(\widehat{\bE}_h,\bm\mu) & = \mathscr{F}(\bm\mu),\notag\\
\mathscr{O}(\theta,\bJ_h) - \mathscr{N}(\theta,U_h) + \mathscr{T}(\widehat{U}_h,\theta) &= 0,\notag
\end{align}
holds for all $(\bm \kappa,\bm\eta,\bm \xi,\zeta,\bm\mu,\theta)\in \bm W_h \times \bm W_h \times \bm W_h \times W_h \times \bm M_h({\bf 0}) \times M_h$. The bilinear forms are given by
\begin{equation*}
\begin{array}{ll}
\mathscr{A}(\bV,\bm \kappa)= (\bV,\bm \kappa)_{\mathcal{T}_h}  ,&
\mathscr{A}_\varepsilon(\bE,\bm \xi)= (\varepsilon_\infty\bE,\bm\xi)_{\mathcal{T}_h},\\
\mathscr{B}(\bE,\bm \kappa)= (\bE,\nabla\times \bm \kappa)_{\mathcal{T}_h},&
\mathscr{C}(\widehat{\bE},\bm \kappa)= \langle\widehat{\bE},\bm \kappa\times\bn\rangle_{\partial \mathcal{T}_h},\\
\mathscr{P}(U,\bm\eta) = (U,\nabla\cdot\bm\eta)_{\mathcal{T}_h} ,&
\mathscr{O}(\widehat{U},\bm\eta) = \langle U,\bm\eta\cdot\bn \rangle_{\mathcal{T}_h} ,\\
\mathscr{D}(\bE,\bm \xi)= \langle\tau_t \bE,\bn\times \bm \xi \times\bn \rangle_{\partial \mathcal{T}_h},&
\mathscr{E}(\widehat{\bE},\bm \xi) = \langle \tau_t\widehat{\bE},\bm \xi\rangle_{\partial \mathcal{T}_h},\\
\mathscr{H}(U,\zeta) = (U,\zeta)_{\mathcal{T}_h} + \langle \tau_nU,\zeta\rangle_{\partial \mathcal{T}_h},&
\mathscr{N}(\widehat{U},\zeta) = \langle \tau_n\widehat{U},\zeta\rangle_{\partial \mathcal{T}_h},\\
\mathscr{R}(\bV,\bm\mu) = \langle \bn\times\bV,\bm\mu\rangle_{\partial \mathcal{T}_h\backslash\partial\overline{\Omega}_E},&
\mathscr{L}(\bE,\bm\mu) = \langle \tau_t\bE,\bm\mu\rangle_{\partial \mathcal{T}_h\backslash\partial\overline{\Omega}_E},\\
\mathscr{M}(\widehat{\bE},\bm\mu) = \langle \widetilde{\tau}_t\widehat{\bE},\bm\mu\rangle_{\partial \mathcal{T}_h},&
\mathscr{T}(\widehat{U},\theta) = \langle \tau_n\widehat{U},\theta\rangle_{\partial \mathcal{T}_h}, \\
\mathscr{F}(\bm\mu)= \langle {\bf f},\bm\mu\rangle_{\partial\overline{\Omega}} .
\end{array}
\end{equation*}
We then discretize the above bilinear forms using the corresponding basis functions on each element/face of the triangulation $\mathcal{T}_h$, and assemble the system of equations that arises from the weak formulation in \eqref{eq:hdghydro_weaksystem}, namely
\begin{equation*}
\left[\begin{array}{cccccc}
\mbb{A} & 0& -\mbb{B} & 0 &-\mbb{C} & 0\\
0& \omega(\omega+i\gamma)\mbb{A} & -i\omega\omega_p^2\mbb{A} & -\beta^2\mbb{P} &0 & \beta^2\mbb{O}\\
\mbb{B}^T & -i\omega\mbb{A} & \mbb{D} -\omega^2\mbb{A}_\varepsilon & 0 & -\mbb{E} &0 \\
0 & -\mbb{P}^T& 0& \mbb{H} & 0 &-\mbb{N}\\
-\mbb{R}& 0& -\mbb{L} & 0 &-\mbb{M} &0\\
0 & -\mbb{O}^T& 0& -\mbb{N}^T & 0 &\mbb{T}\\
\end{array}\right] \left[\begin{array}{c} \underline{\bV} \\ \underline{\bJ} \\ \underline{\bE} \\ \underline{U}\\ \underline{\widehat{\bE}}\\ \underline{\widehat{U}}\end{array} \right] = \left[\begin{array}{c} 0\\ 0\\ 0 \\0 \\ {\bf F}\\ 0\end{array} \right]
\end{equation*}
where $\underline{\bE},\,\underline{\bV},\,\underline{\bJ},\,\underline{U},\,\underline{\widehat{\bE}},\,\underline{\widehat{U}}$ are vectors containing the values of the corresponding fields at the degrees of freedom defined by the discretization $\mcal{T}_h$. The system above, however, is never formed in practice. Instead, we invoke the discontinuity of the approximation spaces to locally eliminate the degrees of freedom of $\bm\Upsilon = (\underline{\bV},\,\underline{\bJ},\,\underline{\bE},\,\underline{U})$, or local unknowns, and express them as a function of only the degrees of freedom of the approximate traces  $\widehat{\bm\Upsilon} = [\underline{\widehat{\bE}},\,\underline{\widehat{U}}]$, or global unknowns. This numerical strategy, also known as hybridization, is paramount to achieve an efficient implementation of the HDG method. The relation between global and local unknowns ${\bm\Upsilon} = \mbb{Z} \widehat{\bm\Upsilon}$, defined at the element level, takes the form
\begin{equation}\label{eq:hdghydro_local}
\left[\begin{array}{c}\underline{\bV} \\ \underline{\bJ} \\ \underline{\bE} \\ \underline{U}\end{array} \right] = \left[\begin{array}{cccc}
\mbb{A} & 0& -\mbb{B} & 0 \\
0& \omega(\omega+i\gamma)\mbb{A} & -i\omega\omega_p^2\mbb{A} & -\beta^2\mbb{P} \\
\mbb{B}^T & -i\omega\mbb{A} & \mbb{D} -\omega^2\mbb{A}_\varepsilon & 0  \\
0 & -\mbb{P}^T& 0& \mbb{H} \\
\end{array}\right]^{-1} \left[\begin{array}{cc}
\mbb{C} & 0\\
0 & -\beta^2\mbb{O}\\
\mbb{E} &0 \\
0 &\mbb{N}\\
\end{array}\right]  \widehat{\bm\Upsilon}, 
\end{equation}
which can be computed efficiently since the matrix is block diagonal, due to the discontinuous nature of the approximation spaces.  The elimination of degrees of freedom through hybridization renders a linear system that involves only the global degrees of freedom, defined at the discretization faces. Hence, we eliminate the local unknowns -- 10 components defined in the high-order volume nodes-- and solve only for the global unknowns -- 3 components defined in the high-order face nodes-- thus drastically reducing the size of the linear system that must be solved. This is one of the most attractive features of the HDG method. Finally, the system involving only the global unknowns is given by 
\begin{equation}
 \lp \left[\begin{array}{cc}
-\mbb{M} &0\\
 0 &\mbb{T}\\
\end{array}\right] +  \left[\begin{array}{cccc}
-\mbb{R}& 0& -\mbb{L} & 0\\
0 & -\mbb{O}^T& 0& -\mbb{N}^T\\
\end{array}\right]\mbb{Z} \rp \widehat{\bm\Upsilon} = \left[\begin{array}{c} {\bf F} \\ 0\end{array} \right].
\end{equation}
This procedure characterizes the solution to \eqref{eq:hdg_hydro} in terms of $\widehat{\bE}_h$ and $\widehat{U}_h$. The local volume variables can be recovered at the element level through \eqref{eq:hdghydro_local}, incurring a small cost as it only involves a matrix-vector product per element, and can be trivially parallelized across elements.

\subsection{Local postprocessing}
We now propose a postprocessing scheme which exploits the superconvergence properties of the HDG method and allows us to recover a more accurate solution in an inexpensive manner. The postprocessed electric and magnetic fields achieve an additional order of convergence in the $\bm H^{\curl}(\mcal{T}_h)$-norm, and according to \cite{cockburn2016hdg} they may be obtained  by solving in each element $T\in\mcal{T}_h$ the following problem
\begin{equation*}
\begin{aligned}
\lp\nabla\times\bE_h^*,\bm \kappa\rp_T&= \lp\bV_h,\bm \kappa\rp_T,\qquad& \forall \bm\kappa& \in  \nabla\times\lb \mcal{P}^{p+1}\lp T\rp \rb^3\;, \\
 \lp \bE_h^*,\bm \xi\rp_T&= \lp \bE_h,\bm \xi\rp_T,\qquad& \forall \bm\xi& \in  \nabla\mcal{P}^{p+2}\lp T\rp  \;, 
\end{aligned} 
\end{equation*}
for the postprocessed electric field $\bE_h^*\in \lb \mcal{P}^{p+1}\lp T\rp \rb^3$, along with
\begin{equation*}
\begin{aligned}
\lp \nabla\times\bV_h^*,\bm \kappa\rp_T&= \lp \omega^2\varepsilon_\infty\bE_h + i\omega\bJ_h,\bm \kappa\rp_T,\qquad& \forall \bm\kappa& \in  \nabla\times\lb \mcal{P}^{p+1}\lp T\rp \rb^3\;, \\
 \lp \bV_h^*,\bm \xi\rp_T&= \lp \bV_h,\bm \xi\rp_T,\qquad& \forall \bm\xi& \in   \nabla \mcal{P}^{p+2}\lp T\rp  \;, 
\end{aligned} 
\end{equation*}
for the postprocessed magnetic field $\bV_h^*\in \lb \mcal{P}^{p+1}\lp T\rp \rb^3$. The curl of $\bE_h^*,\,\bV_h^*$ corresponds to projections onto the subspace of functions in $\lb \mcal{P}^{p+1}\lp T\rp \rb^3$ with zero divergence, hence we expect a $p+1$ convergence rate for the postprocessed variables in $\bm H^{\curl}\lp \mcal{T}_h\rp$-norm.

Similarly, the electric current may be postprocessed by solving
\begin{equation*}
\begin{aligned}
\lp\nabla\cdot\bJ_h^*,\zeta \rp_T&= \lp U_h,\zeta \rp_T,\qquad& \forall \zeta & \in  \mcal{P}^{p+1}\lp T\rp \;, \\
 \lp \bJ_h^*,\bm \xi\rp_T&= \lp \bJ_h,\bm \xi\rp_T,\qquad& \forall \bm\xi& \in  \lb \mcal{P}^{p+1}\lp T\rp \rb^3  \;,
\end{aligned} 
\end{equation*}
where $\bJ_h^*\in \lb \mcal{P}^{p+1}\lp T\rp \rb^3$ achieves a $p+1$ convergence rate in the $\bm H^{\dive}(\mcal{T}_h)$-norm. Finally, in order to postprocess the additional variable $U_h$, we recall that $\nabla U_h$ can be computed locally by virtue of the third equation in \eqref{eq:firstordersystemhydro}. Hence, we can recover a postprocessed $U_h^*\in\mcal{P}^{p+1}\lp T\rp$ element-wise solving
\begin{equation*}
\begin{aligned}
\lp \nabla U_h^*, \nabla \zeta\rp_T&= \frac{1}{\beta^2}\lp i\omega\omega_p^2 \bE_h - \omega\lp \omega+i\gamma\rp\bJ_h,\nabla \zeta\rp_T,\qquad& \forall \zeta& \in   \mcal{P}^{p+1}\lp T\rp\;, \\
 \lp U_h^*,1\rp_T&= \lp U_h,1\rp_T,\qquad&   \;.
\end{aligned} 
\end{equation*}
which is shown to converge at the rate of $p+2$.

The main advantage of this approach is that the postprocessed approximate fields $(\bV_h^*,\bE_h^*,\bJ_h^*,U_h^*)$ are significantly less expensive to obtain than the original approximate fields $(\bV_h,\bE_h,\bJ_h,U_h)$, as its computation does not involve the solution of any global system. Furthermore, each variable is independently postprocessed at the element level, hence the linear systems above are much smaller than the linear system \eqref{eq:hdghydro_local} required to assemble the global system during hybridization. In addition, local postprocessing is an embarrassingly parallel task. It can therefore be concluded that postprocessing the local variables will have a minor impact in the overall computational cost.

\subsection{Metal-dielectric coupling}
In this section, we examine the scenario where a metal $\overline{\Omega}$, described by the hydrodynamic model, is embedded in a dielectric medium $\Omega$ with permittivity $\varepsilon_d$. Consider, for instance, a metallic nanostructure, surrounded by a dielectric medium, scattering an incident ${\bf p}$-polarized plane wave $\bE_0$ propagating in the ${\bf d}$-direction, that is ${\bf E}_0 = {\bf p}\exp(i\omega\sqrt{\varepsilon_d}\,{\bf d}\cdot \bx)$, as shown in Fig. \ref{fig:coupling} (left).

\begin{figure}[h!]
 \centering
 \includegraphics[scale = 1]{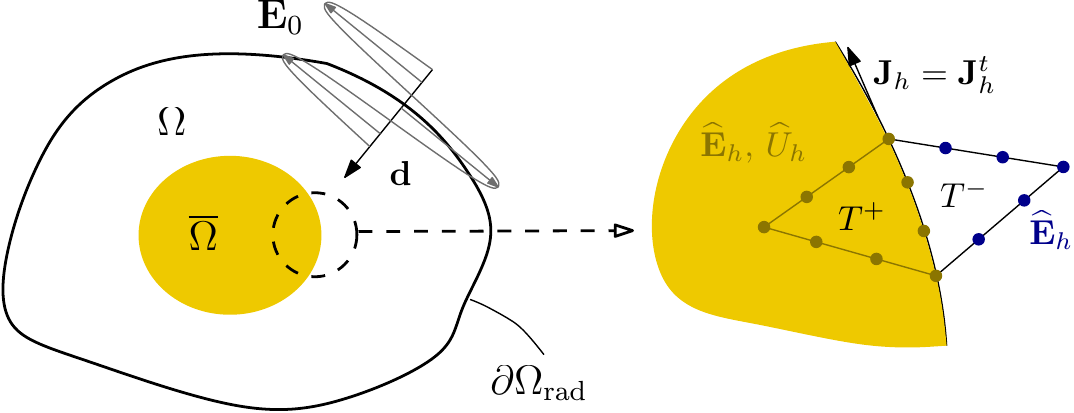}
 \caption{Left: Metallic structure $\overline{\Omega}$ embedded in dielectric $\Omega$ illuminated by plane wave. Right: Detail of metal-dielectric interface with global degrees of freedom.}\label{fig:coupling}
\end{figure}

In this situation, there are two subdomains with different governing equations. The solution within the metallic structure is governed by
\begin{align*}
\overline{\mcal{L}} &= 0,\quad \mbox{in } \overline{\Omega},\\
\bJ \cdot \bn & = 0,\quad \mbox{on }\partial \overline{\Omega},
\end{align*}
whose HDG discretization is given by \eqref{eq:hdg_hydro}. Note that the no electron spill-out condition enforces that the electric current at the metallic interface is tangential $\bJ_h = \bJ_h^t$.

Conversely, the response in the dielectric $\Omega$ is given by regular time-harmonic Maxwell's equations, namely
\begin{align*}
\nabla \times \bE - \bV &= 0 , \\
\nabla \times \bV - \omega^2\varepsilon_d \bE & = 0 ,
\end{align*}
The HDG discretization of the time-harmonic Maxwell's equations is described in detail in \cite{nguyen2011maxwell}. The boundaries of the surrounding dielectric medium $\Omega$ represent the farfield truncation of the infinite space, where radiation is imposed with the Silver-M\"{u}ller conditions, which are first order absorbing boundary conditions \cite{sommerfeld1949partial,mur1981absorbing}, namely 
\begin{equation}\label{eq:silvermuller}
(\bV- \bV_0)\times \bn -  i\omega\sqrt{\varepsilon_d}\,\bn\times(\bE-\bE_0)\times\bn = 0,\quad \mbox{on }\partial\Omega_{\rm{rad}}\,.
\end{equation}

Finally, we need to impose a compatibility condition to stitch the subdomains together. For any two elements $T^+,\,T^-$ that satisfy $T^+\cap T^- \in \partial \overline{\Omega}$, see Fig. \ref{fig:coupling} (right), we enforce continuity of the tangential component of the trace of the magnetic field $\llbracket \bn\times\widehat{\bV}_h \rrbracket =0$ at the interface. Furthermore, since the traces are single-valued across inter-element boundaries, the global degrees of freedom on the faces $F\in \partial T^-$ have two $\lbrace \widehat{\bE}_h\rbrace$ (resp. three $\lbrace\widehat{\bE}_h,\,\widehat{U}_h\rbrace$) components for $F \notin \partial \overline{\Omega}$ (resp. $F \in \partial \overline{\Omega}$). Thus, the assembly of the global matrix needs to account for the global compatibility condition and the different number of global components.

\section{Numerical results}\label{sec:res}
\subsection{Convergence test}
In this section, we perform numerical tests to examine the convergence and accuracy of the HDG method for the HM introduced above. To that end, we solve \eqref{eq:firstordersystemhydro} in a square domain $\overline{\Omega} = (0,\pi)^2$ with $\varepsilon_\infty = 2$. In addition, we set $\omega = \omega_p  =1$, $\gamma=0$ and $\beta^2  = 0.5$ and select boundary data $\bE_\partial$ and $\bn\cdot\bJ$ such that the problem has the following exact solution
\begin{alignat*}{2}
 \bE &= (\cos x - i\sin y,\cos y - i\sin x),\qquad \quad \bV  = i\cos y-i\cos x,\\
 \bJ &= (\sin y + 2i\cos x,\sin x + 2i\cos y),\qquad  \uprho  = -2\sin x - 2\sin y  .
\end{alignat*}
The stabilization parameters are set according to the values proposed above, that is $\tau_t = \tau_n = \sqrt{2}$. We  analyze the convergence of the method on a sequence of structured triangular meshes $\mcal{T}_h$ with $n^2/2$ elements by computing the $\bm L^2(\mcal{T}_h),\,\bm H^{\curl}(\mcal{T}_h)$ and $\bm H^{\dive}(\mcal{T}_h)$ norm of the errors for the above variables.

\begin{table}[h!]
\footnotesize
\begin{center}
  \renewcommand{\arraystretch}{1.2}
\begin{tabular}{cc|cccccccccc}
 &   & \multicolumn{2}{c}{$\norm{\bE- \bE_h}_{\bm L^2}$} & \multicolumn{2}{c}{$\norm{\bE- \bE_h}_{\bm H^{\curl}}$} &\multicolumn{2}{c}{$\norm{\bJ- \bJ_h}_{\bm L^2}$} &\multicolumn{2}{c}{$\norm{\bJ- \bJ_h}_{\bm H^{\dive}}$}&\multicolumn{2}{c}{$\norm{\uprho- \uprho_h}_{\bm L^2}$}\\
$p$& $n$  & Error& Order& Error& Order& Error& Order & Error& Order & Error& Order\\
[1mm]
1 & 8 & 3.6e-2 & --   & 3.8e-1 & --   & 6.7e-2 & --   & 1.1e0 & --   & 3.9e-2 & -- \\
  &16 & 8.6e-3 & 2.06 & 1.9e-1 & 1.02 & 1.5e-2 & 2.15 & 5.3e-1 & 1.07 & 5.5e-3 & 2.84 \\
  &32 & 2.1e-3 & 2.02 & 9.5e-2 & 1.00 & 3.6e-3 & 2.05 & 2.6e-1 & 1.03 & 9.3e-4 & 2.56 \\
  &64 & 5.3e-4 & 2.01 & 4.7e-2 & 1.00 & 8.9e-4 & 2.02 & 1.3e-1 & 1.01 & 2.0e-4 & 2.24 \\
  [2mm]
2 & 8 & 1.1e-3 & --   & 1.9e-2 & --   & 1.8e-3 & --   & 5.5e-2 & -- & 4.7e-4 & -- \\
  &16 & 1.3e-4 & 3.04 & 4.7e-3 & 2.02 & 2.2e-4 & 3.07 & 1.3e-2 & 2.04 & 5.6e-5 & 3.08 \\
  &32 & 1.6e-5 & 3.02 & 1.2e-3 & 2.01 & 2.7e-5 & 3.03 & 3.3e-3 & 2.02 & 6.9e-6 & 3.02 \\
  &64 & 2.0e-6 & 3.01 & 2.9e-4 & 2.00 & 3.3e-6 & 3.01 & 8.2e-4 & 2.01 &  8.6e-7& 3.00 \\
  [2mm]
  3 & 8 & 2.7e-5 & --   & 6.8e-4 & --   & 4.7e-5 & --   & 2.1e-3 & --   & 1.3e-5 & -- \\
  &16 & 1.7e-6 & 4.01 & 8.5e-5 & 3.01 & 2.9e-6 & 4.05 & 2.6e-4 & 3.04 & 7.9e-7 & 4.01 \\
  &32 & 1.1e-7 & 4.00 & 1.1e-5 & 3.00 & 1.8e-7 & 4.02 & 3.2e-5 & 3.02 & 4.9e-8 & 4.00 \\
  &64 & 6.6e-9 & 4.00 & 1.3e-6 & 3.00 & 1.1e-8 & 4.01 & 4.0e-6 & 3.01 & 3.1e-9 & 4.00 \\
  
\end{tabular}
\end{center}
\caption{History of convergence for the approximate solution.}\label{tab:hdghydro}
\end{table}

We consider polynomials of degree $p = 1,\,2$ and $3$ to represent the solution, and present the results in Table \ref{tab:hdghydro} for the approximate solutions and in Table \ref{tab:hdghydro_post} for the postprocessed solutions. We observe that both the electric field, electric current and induced free charge converge at the optimal rate of $\mcal{O}(h^{p+1})$ in the $\bm L^2(\mcal{T}_h)$-norm, whereas the electric field (resp. electric current) converges at the rate of $\mcal{O}(h^{p})$ in the $\bH^{\curl}(\mcal{T}_h)$-norm (resp. $\bH^{\dive}(\mcal{T}_h)$-norm).

\begin{table}[h!]
\footnotesize
\begin{center}
  \renewcommand{\arraystretch}{1.2}
\begin{tabular}{cc|cccccccccc}
 &   & \multicolumn{2}{c}{$\norm{\bE- \bE^*_h}_{\bm L^2}$} & \multicolumn{2}{c}{$\norm{\bE- \bE^*_h}_{\bm H^{\curl}}$} &\multicolumn{2}{c}{$\norm{\bJ- \bJ^*_h}_{\bm L^2}$} &\multicolumn{2}{c}{$\norm{\bJ- \bJ^*_h}_{\bm H^{\dive}}$}&\multicolumn{2}{c}{$\norm{\uprho- \uprho^*_h}_{\bm L^2}$}\\
$p$& $n$  & Error& Order& Error& Order& Error& Order & Error& Order & Error& Order\\
[1mm]
1 & 8 & 3.9e-2 & --   & 4.6e-2 & --   & 6.9e-2 & --   & 8.5e-2 & --   & 3.8e-2 & -- \\
  &16 & 9.4e-3 & 2.06 & 1.1e-2 & 2.10 & 1.5e-2 & 2.15 & 1.8e-2 & 2.22 & 4.7e-3 & 3.02 \\
  &32 & 2.3e-3 & 2.02 & 2.6e-3 & 2.03 & 3.7e-3 & 2.05 & 4.3e-3 & 2.07 & 5.9e-4 & 3.01 \\
  &64 & 5.8e-4 & 2.01 & 6.5e-4 & 2.01 & 9.1e-4 & 2.02 & 1.1e-3 & 2.02 & 7.3e-5 & 3.01 \\
[2mm]
2 & 8 & 1.1e-3 & --   & 1.3e-3 & --   & 1.7e-3 & --   & 2.0e-3 & -- & 1.7e-4 & -- \\
  &16 & 1.4e-4 & 3.04 & 1.6e-4 & 3.03 & 2.0e-4 & 3.07 & 2.4e-4 & 3.06 & 8.6e-6 & 4.31 \\
  &32 & 1.7e-5 & 3.02 & 2.0e-5 & 3.01 & 2.5e-5 & 3.03 & 3.0e-5 & 3.02 & 4.9e-7 & 4.12 \\
  &64 & 2.1e-6 & 3.01 & 2.4e-6 & 3.00 & 3.1e-6 & 3.01 & 3.7e-6 & 3.01 &  3.0e-8& 4.05 \\
[2mm]
  3 & 8 & 2.9e-5 & --   & 3.2e-5 & --   & 4.5e-5 & --   & 5.2e-5 & --   & 2.9e-6 & -- \\
  &16 & 1.8e-6 & 4.01 & 2.0e-6 & 4.01 & 2.7e-6 & 4.04 & 3.2e-6 & 4.03 & 8.3e-8 & 5.12 \\
  &32 & 1.1e-7 & 4.01 & 1.3e-7 & 4.00 & 1.7e-7 & 4.02 & 2.0e-7 & 4.01 & 2.5e-9 & 5.05 \\
  &64 & 6.9e-9 & 4.00 & 7.8e-9 & 4.00 & 1.0e-8 & 4.01 & 1.2e-8 & 4.01 & 7.7e-11 & 5.02 \\
\end{tabular}
\end{center}
\caption{History of convergence for the postprocessed solution.}\label{tab:hdghydro_post}
\end{table}

Nonetheless, the local postprocessing described above recovers an additional order of convergence $p+1$ on the $\bm H^{\curl}(\mcal{T}_h)$-norm for the electric field and on the $\bm H^{\dive}(\mcal{T}_h)$-norm for the electric current, as well as an optimal convergence rate of $p+2$ for the induced free charge $\uprho$.

\subsection{Single cylindrical nanowire}
In order to show the differences between the LRA and the HM, we consider a golden nanowire of diameter $D$ in free space. We assume the nanowire is infinite in the $z$ direction, and is excited by an $x$-polarized electric field propagating the $y$-direction, that is $\bE_0 = \exp(i\omega y)\hat{\bm x}$, see Fig. \ref{fig:nanowire}. For this simple geometry, the analytical solution is available for both the LRA and the HM using Bessel and Hankel functions \cite{ruppin2001extinction}, and is useful to illustrate the physics captured by both models. The quantity of interest is the extinction cross section
\begin{equation*}
 \ext = -\dfrac{1}{D\abs{\bE_0}^2}\int_A \Re\lb\bE_0\times\bH^* + \bE\times\bH^*_0\rb\cdot d{\bf A}
\end{equation*}
where $A$ is an arbitrary area enclosing the wire. Results are computed for both local and nonlocal models, with diameters 4 and 40 nm. The values for the gold constants are $\varepsilon_\infty = 1$, $\hslash\overline{\omega}_p = 9.02$ eV and $\hslash \overline{\gamma} = 0.071$ eV \cite{johnson1972optical}, where $\hslash = h/2\pi$ is the reduced Planck constant, and $\widebar{v}_F = 1.39\cdot 10^6$ m/s \cite{ashcroft2005solid}.

 \begin{figure}[h!]
\centering
\subcaptionbox{\label{fig:nanowire}}
[6cm]{\includegraphics[scale = 0.4]{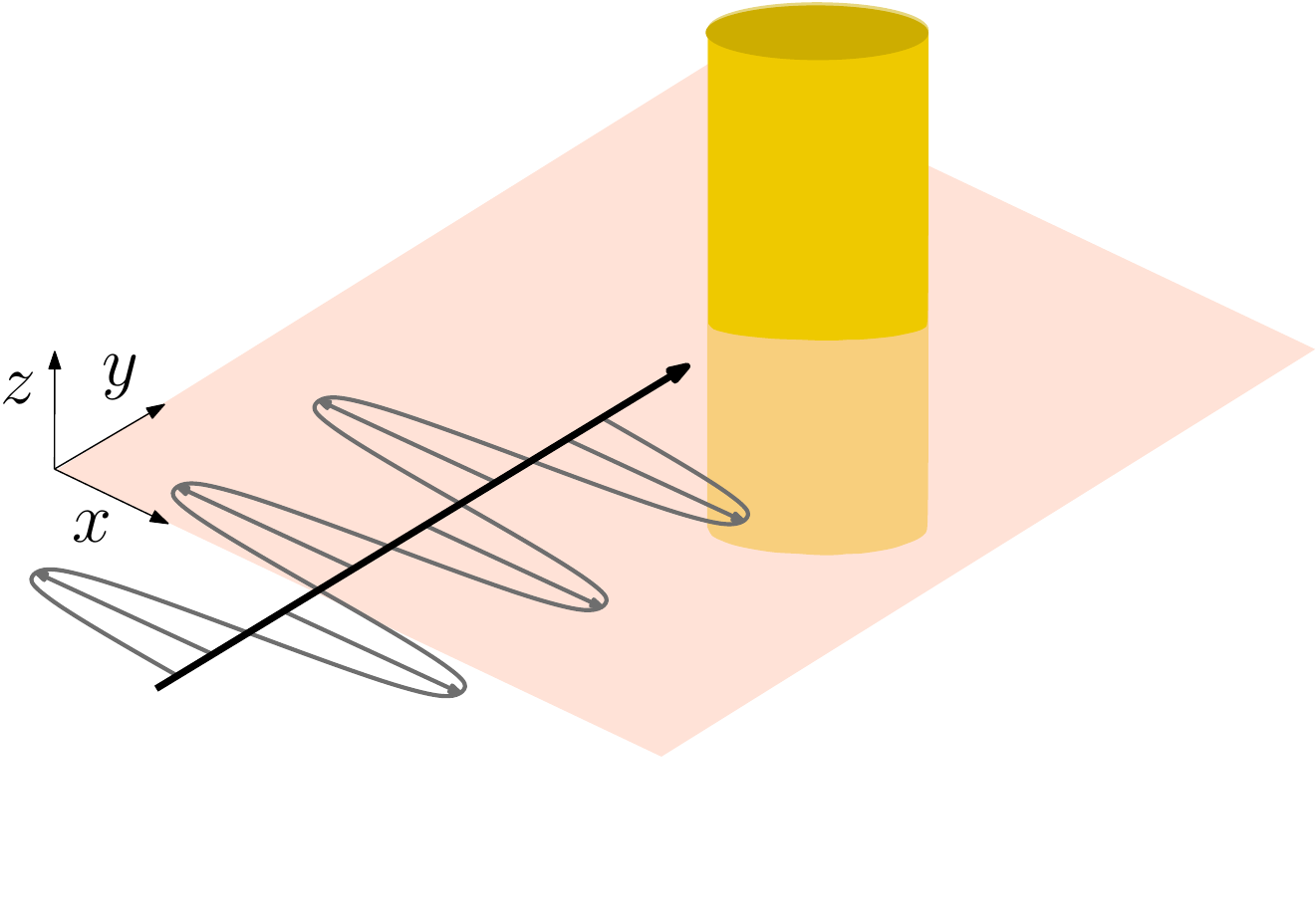}}
\hfill\subcaptionbox{\label{fig:wiremesh}}
[8.75cm]{\includegraphics[scale = 0.45]{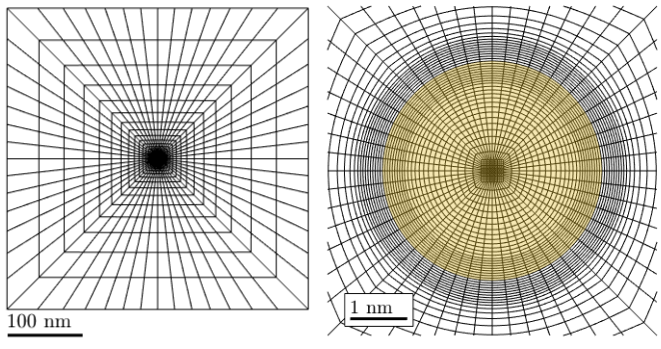}}
\caption{(a) Schematic diagram of single nanowire under plane wave illumination. (b) Two views of the cubic discretization, with gold nanowire highlighted.}
\end{figure}
For this simulation, we set the computational domain to be a square of 0.4 \mm $\times$ 0.4 \mm, and prescribe Silver-M\"{u}ller conditions on the  boundaries. The size of the computational domain is chosen such that the location of the radiating boundaries is far enough so that it has no significant effect on the solution. The domain is discretized with an anisotropic mesh of 3600 cubic quadrangular elements, ensuring that greater resolution is achieved near the nanoparticle, see Fig. \ref{fig:wiremesh}, with element sizes ranging from 50 nm to 0.05 nm. The theoretical results given in \cite{ruppin2001extinction} are visually indistinguishable from the numerical ones, with relative errors below 1\% for all frequencies.

\begin{figure}[h!]
 \centering
 \includegraphics[scale = .4]{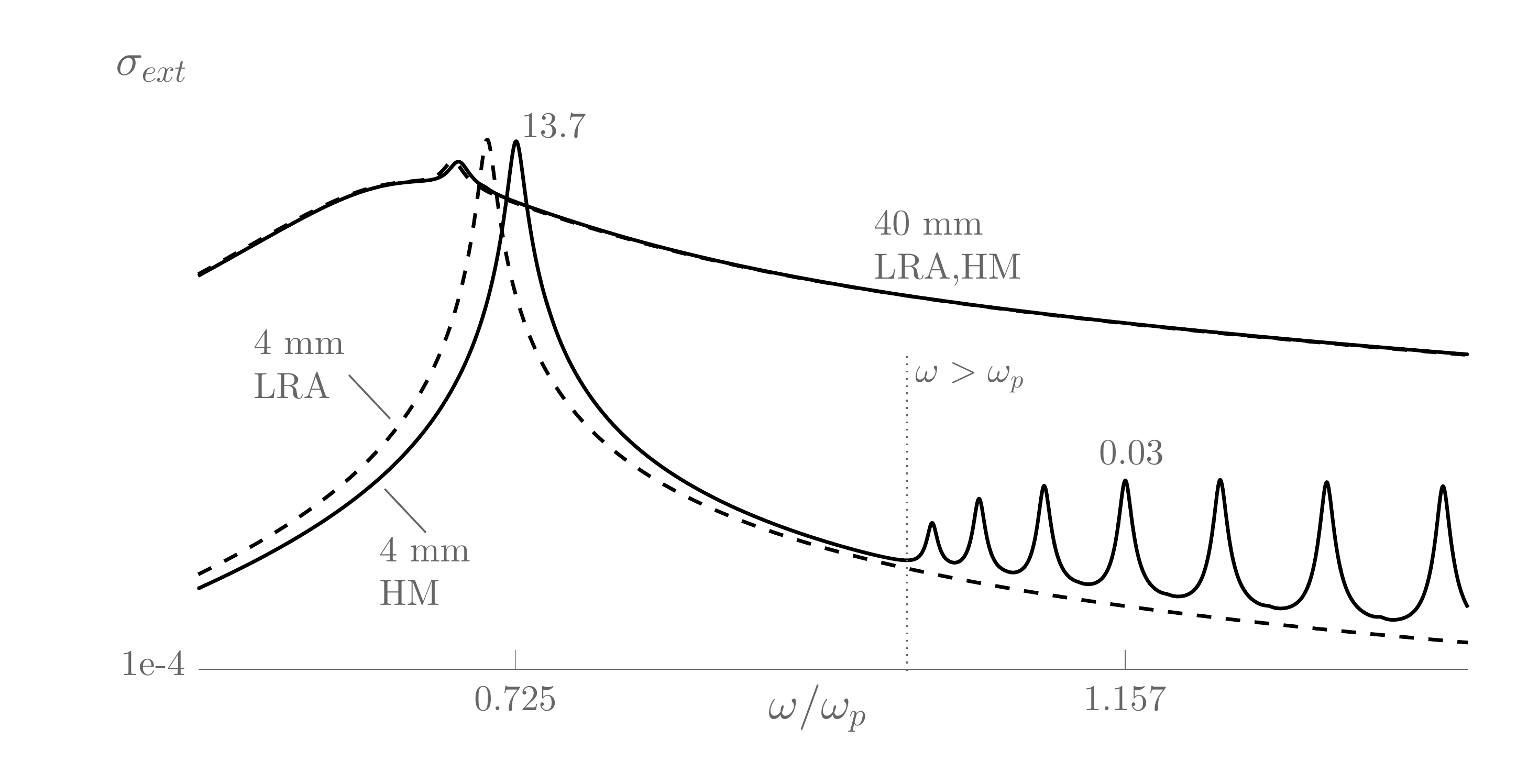}
 \caption{Extinction cross section $\ext$ (logscale) of gold nanowire with diameters 4 and 40 nm for LRA and HM. Nonlocal effects modeled with the HM are only relevant at nanometric scales, blue-shifting the main resonance and exciting volume modes above the plasma frequency.}\label{fig:nanowiresigma}
\end{figure}


As anticipated, for small metallic nanoparticles the effects of the hydrodynamic current are significant, causing not only a blue-shift of around 3\% in the main resonance, but also a sequence of resonances above the plasma frequency that are not excited with the local model, see Fig. \ref{fig:nanowiresigma}. These excitations correspond to volume plasmon states, which are confined longitudinal oscillations of the electron gas. It can be shown \cite{ciraci2013hydrodynamic} that below the plasma frequency both the transverse and the longitudinal modes decay exponentially, whereas above the plasma frequency both modes propagate. As a matter of fact, it is the propagation of the longitudinal modes that causes the additional resonances shown in Fig. \ref{fig:nanowiresigma} for $\omega>\omega_p$. Conversely, the more simplistic local model only allows a longitudinal mode at the plasma frequency.

Furthermore, results for the 40 nm wire show that the hydrodynamic model predicts a response very similar to that of the LRA. Hence, including the hydrodynamic pressure term is only relevant for nanometric geometries.

 \begin{figure}[h!]
\subcaptionbox{\label{fig:r2local}}
[4.5cm]{\includegraphics[scale = 0.3]{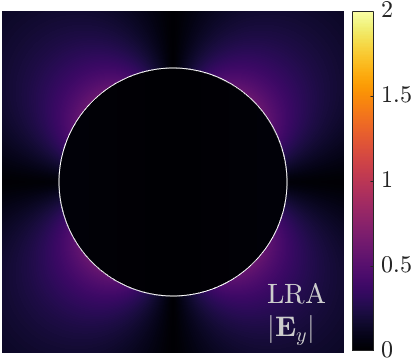}}
\hfill\subcaptionbox{\label{fig:r2nonlocal}}
[4.5cm]{\includegraphics[scale = 0.3]{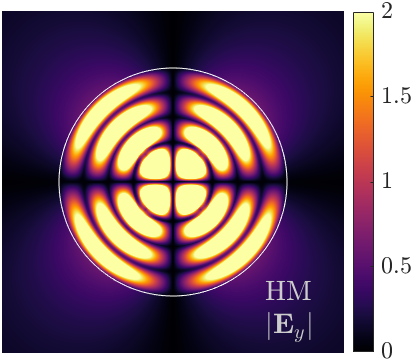}}
\hfill\subcaptionbox{\label{fig:r2J}}
[4.5cm]{\includegraphics[scale = 0.3]{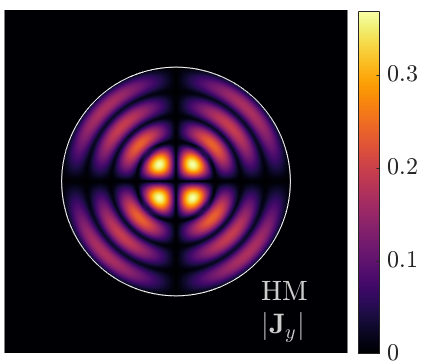}}
\caption{Solution fields for 4 nm gold wire at $\omega/\omega_p = 1.157$, with boundary highlighted.}\label{fig:nanowiresol}
\end{figure}

We shall now inspect the solutions of the EM field for the 4 nm wire. The inclusion of the electron pressure term excites features that occur at the sub-Fermi-wavelength scale. This wavelength is associated to the Fermi energy --the maximum energy of the electrons in the metal-- and is typically much smaller than the length scale of the problem. Indeed, the field $\abs{{\bf E}_y}$ for both models at the resonant frequency $\omega/\omega_p = 1.157$ shown in Fig. \ref{fig:nanowiresol} illustrates this phenomenon. Even though the solution outside the metal is similar, modeling the hydrodynamic current results in wave patterns inside the nanowire of wavelength 100 times smaller than the wavelength of the incident field, see $\abs{{\bf E}_y}$ in Figs \ref{fig:r2nonlocal} and $\abs{{\bf J}_y}$ in \ref{fig:r2J}, due to the excitation of a longitudinal plasmon. Consequently, to properly capture the nonlocal effects predicted by the HM we require significantly finer discretizations in the metallic structure.

\subsection{3D periodic annular nanogap}
We now consider a 3D structure, the periodic annular nanogap, which has been shown to produce extraordinary optical transmission and enormous field enhancements \cite{rodrigo2016extraordinary,park2015nanogap,yoo2016high}. These structures consist of periodic arrays of subwavelength annular apertures of a dielectric material patterned in a metallic film, and unlike arrays of circular and rectangular apertures they sustain plasmon resonances for a broad range of frequency regimes. That is, for a fixed gap size one can adjust the ring diameter and the array periodicity to generate resonances for the visible, the mid infrared (MIR), the far infrared (FIR) regime and the terahertz (THz) regime. 

Researchers have demonstrated high-throughput fabrication schemes to make nanometer-wide annular gaps with perimeters of microns to millimeters \cite{park2015nanogap,yoo2016high,chen2013atomic,im2010vertically,chen2014squeezing}. Such resonant nanogap structures have been used for plasmonic sensing applications as well as fundamental studies of nanophotonics phenomena. These technological advances motivate fast numerical modeling of such extreme-scale 3D structures, consisting of sub-10 nm-gap annular apertures with micron- to millimeter-scale diameters.



\begin{figure}[h!]
\centering
\subcaptionbox{\label{fig:annularDiameter}}
[15cm]{\includegraphics[scale = 0.185]{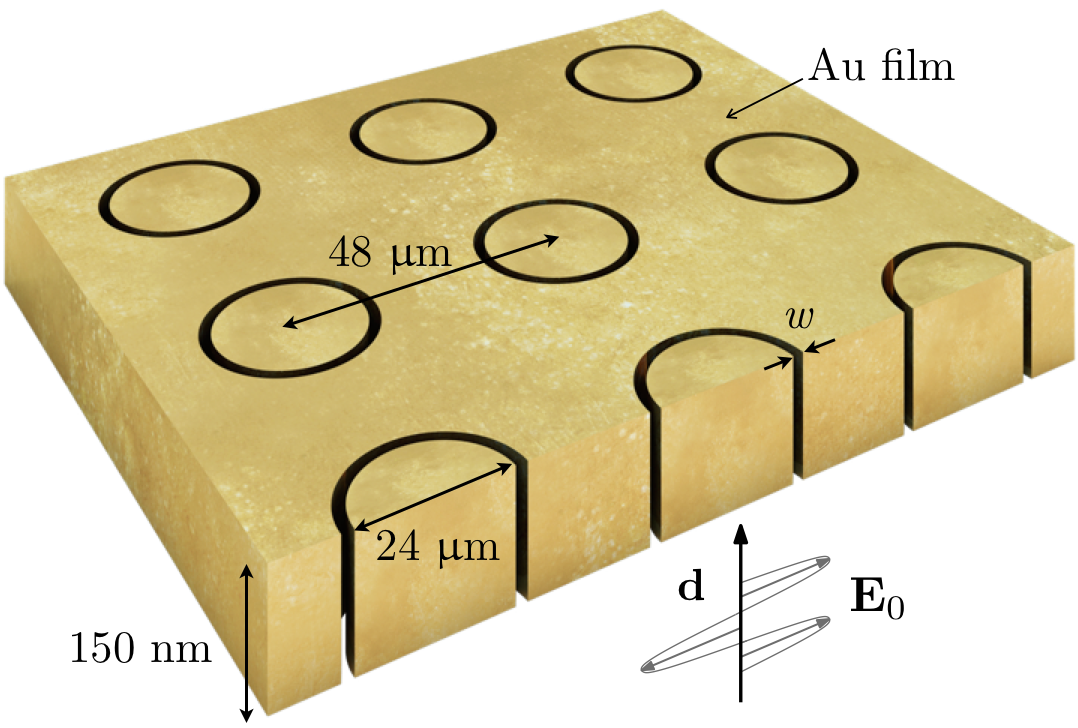}}
\hfill\subcaptionbox{\label{fig:annularFIR}}
[8cm]{\includegraphics[scale = 0.48]{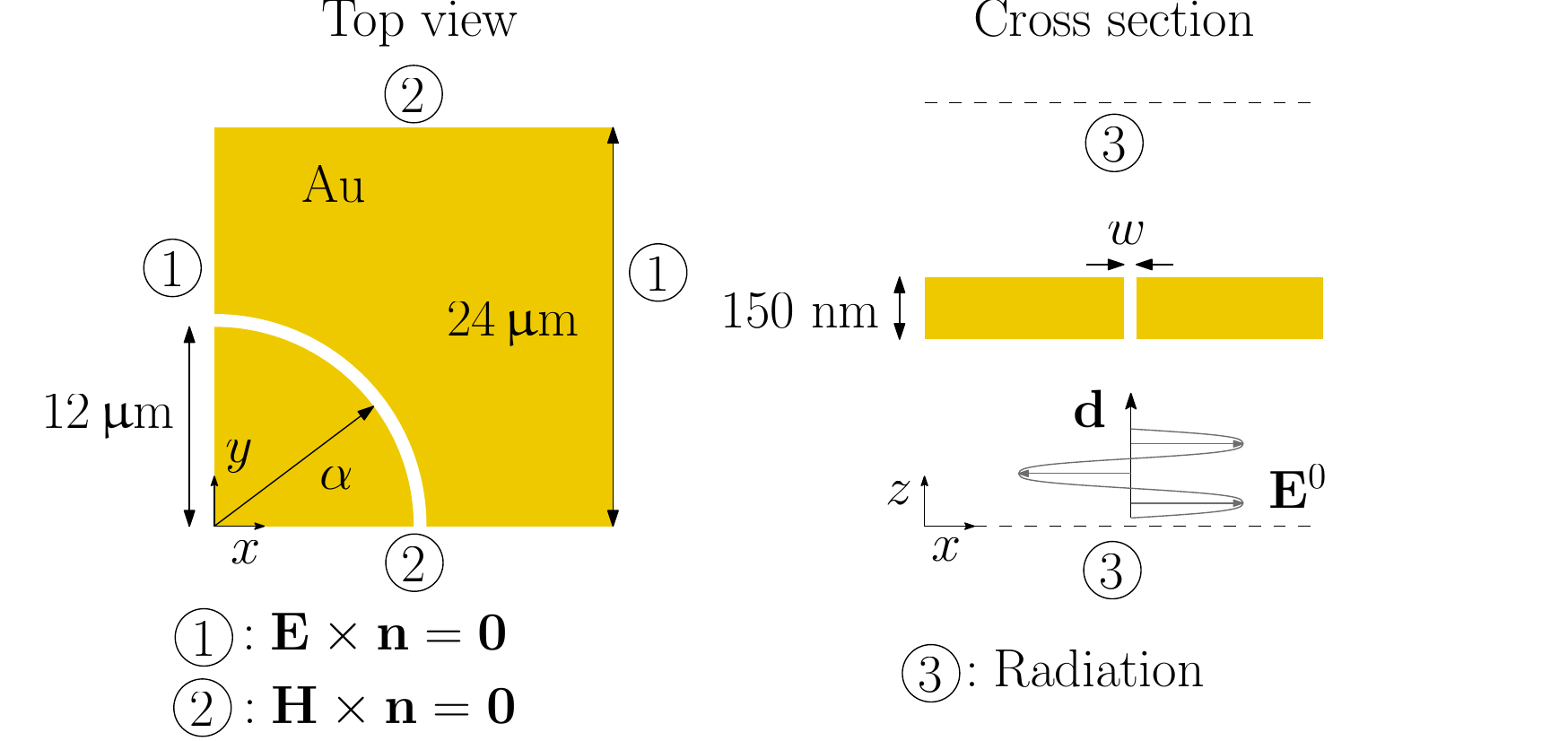}}
\hfill\subcaptionbox{\label{fig:detail_mesh}}
[7cm]{\includegraphics[scale = 0.48]{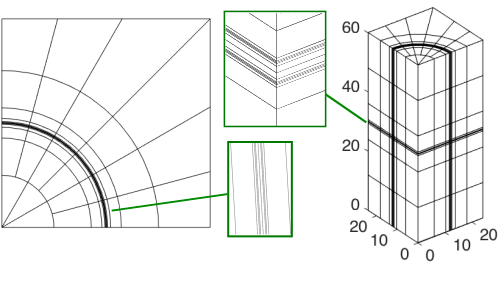}}
\caption{(a) Schematic of periodic array of annular gaps with relevant dimensions. (b)  Top and cross section view of unit computational cell for periodic annular array. (c)  3D high-order mesh and 2D slice (with details) used in calculations.}
\end{figure}

The structure that will be analyzed is a gold thin-film with annular nanogaps arranged according to the symmetries of the square, see Fig. \ref{fig:annularDiameter}. In order to focus only on the impact of the metal, we shall assume the film is suspended in free space (no substrate), and that there is no material filling the nanometer-wide gap. Although this structure cannot be manufactured, it is of interest to achieve a deeper understanding of the ring structure from a theoretical perspective. The structure is illuminated from below with an $x$-polarized plane wave ${\bf E}_0 = \exp(i\omega z)\hat{\bm x}$ with frequencies in the low THz regime. Note that the problem may be further reduced by exploiting the symmetries of the lattice, hence we only need to solve for one quadrant of the ring structure as indicated in Fig. \ref{fig:annularFIR}. Symmetry is enforced, for an $x$-polarized plane wave, by imposing $\bE\times\bn = \bm 0$ on the $x$-constant  boundaries and $\bH\times\bn = \bm 0$ on the $y$-constant  boundaries. Radiation conditions \eqref{eq:silvermuller} are imposed on the $z$-constant boundaries.

We consider aperture widths $w$ ranging from 0.5 nm to 100 nm, for frequencies between 0.2 THz and 5.5 THz, and investigate the response using the distinct models for light-metal interaction introduced above. More specifically, the outputs monitored are the transmitted power $\pow$ through the structure and the enhancement $\pi$ of the $x$-component of the electric field in the gap volume, computed as
\begin{equation}
  \pow = \frac{\int_{A_1} \abs{\Re \lb \bE \times \bH^*\rb \cdot \bn} d{ A}}{\int_{A_0} \abs{\Re \lb \bE_0 \times \bH_0^*\rb \cdot \bn} d{ A}}\,,\qquad \qquad
\fie = \frac{\int_{\textnormal{gap}} \abs{\bE_x} dV}{\int_{\textnormal{gap}} \abs{\bE_{0,x}} dV}\,, 
\end{equation}
where $A_0$ is an arbitrary $xy$ plane below the gold film and $A_1$ an arbitrary $xy$ plane above the gold film. In this frequency regime, the 3D periodic annular nanogap excites resonances whose electric field's $x$-component is constant along the aperture, thus we focus only on the enhancement of this component.
 
The discretization consists of 1.8K hexahedral cubic elements, and is constructed by extruding in the $z$-direction the 2D curved mesh in Fig. \ref{fig:detail_mesh}, with the inset showing the concentration of elements in the vicinity of the gap. In addition, we also present the entire 3D mesh, along with an inset showing a zoom of the gold film region. The hexahedral elements in the vertical direction are smaller close to the upper and lower surfaces of the gold film, and gradually increase as we separate from the metal. The radiation conditions are prescribed at 30 microns for the glass substrate and 30 microns for air, ensuring there is no numerical interaction between the boundary and the extraordinary optical transmission that occurs in the ring. This highly anisotropic mesh allows us to solve for the full 3D EM wave field using a reduced number of degrees of freedom. The numerical accuracy is verified by carrying out grid convergence studies on consecutively refined meshes, until the relative error for the field enhancement of the smallest gap is below 0.1\%. We then perform, for each electron model and gapsize under consideration, 5000 3d HDG simulations at different frequencies within the interval of interest. These frequency sweeps give rise to the $\pi$ and $\varsigma$ profiles presented in Fig. \ref{fig:ringTheory}, and enable the tracking of the resonance for each case.


\begin{figure}[h!]
\centering
\subcaptionbox{\label{fig:ringTheory_PEC}}
[7.5cm]{\includegraphics[scale = 0.35]{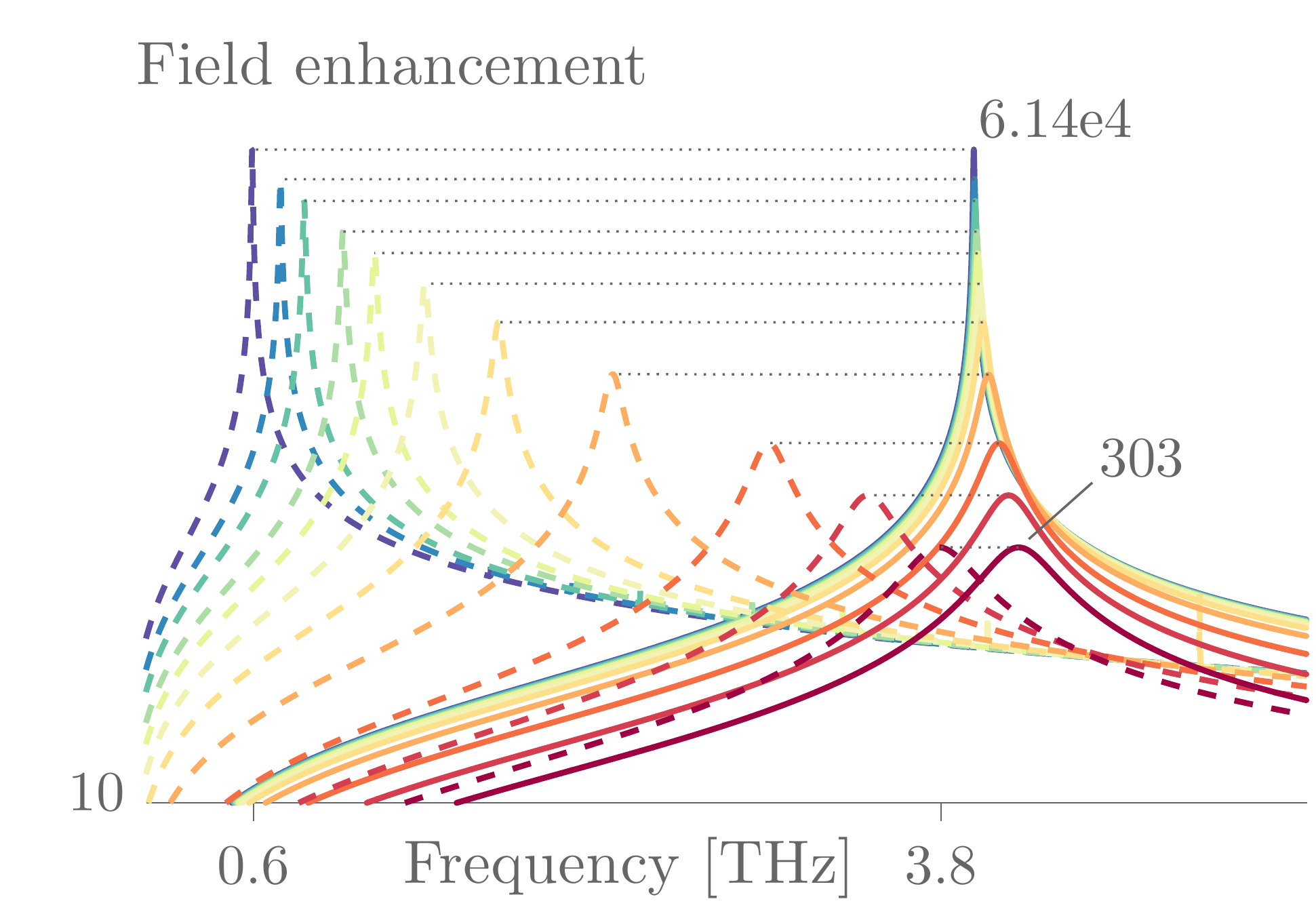}}
\hfill\subcaptionbox{\label{fig:ringFE_theory}}
[7.5cm]{\includegraphics[scale = 0.35]{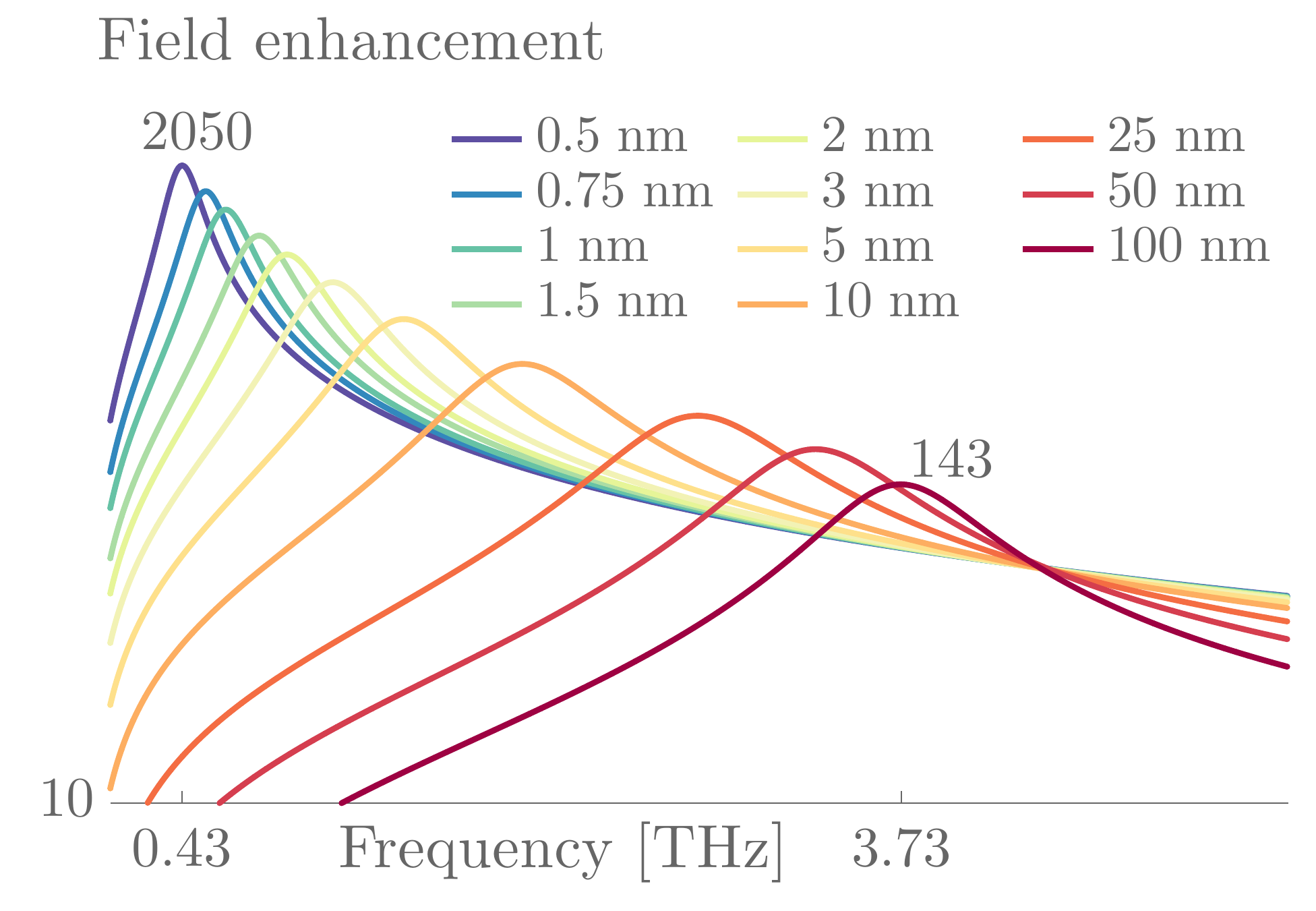}}
\hfill\subcaptionbox{\label{fig:ringTP_theory}}
[7.5cm]{\includegraphics[scale = 0.35]{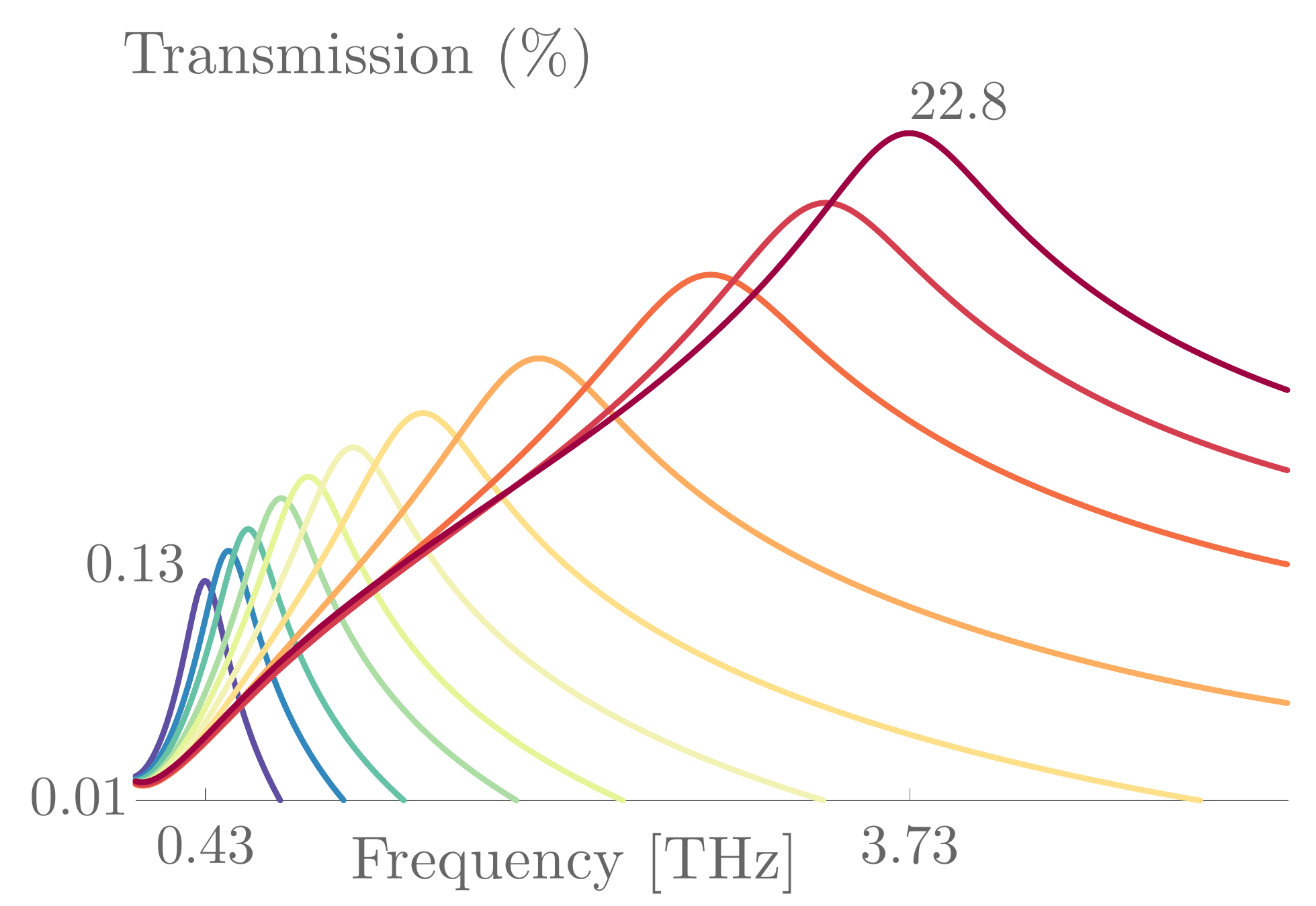}}
\hfill\subcaptionbox{\label{fig:ringATP_theory}}
[7.5cm]{\includegraphics[scale = 0.35]{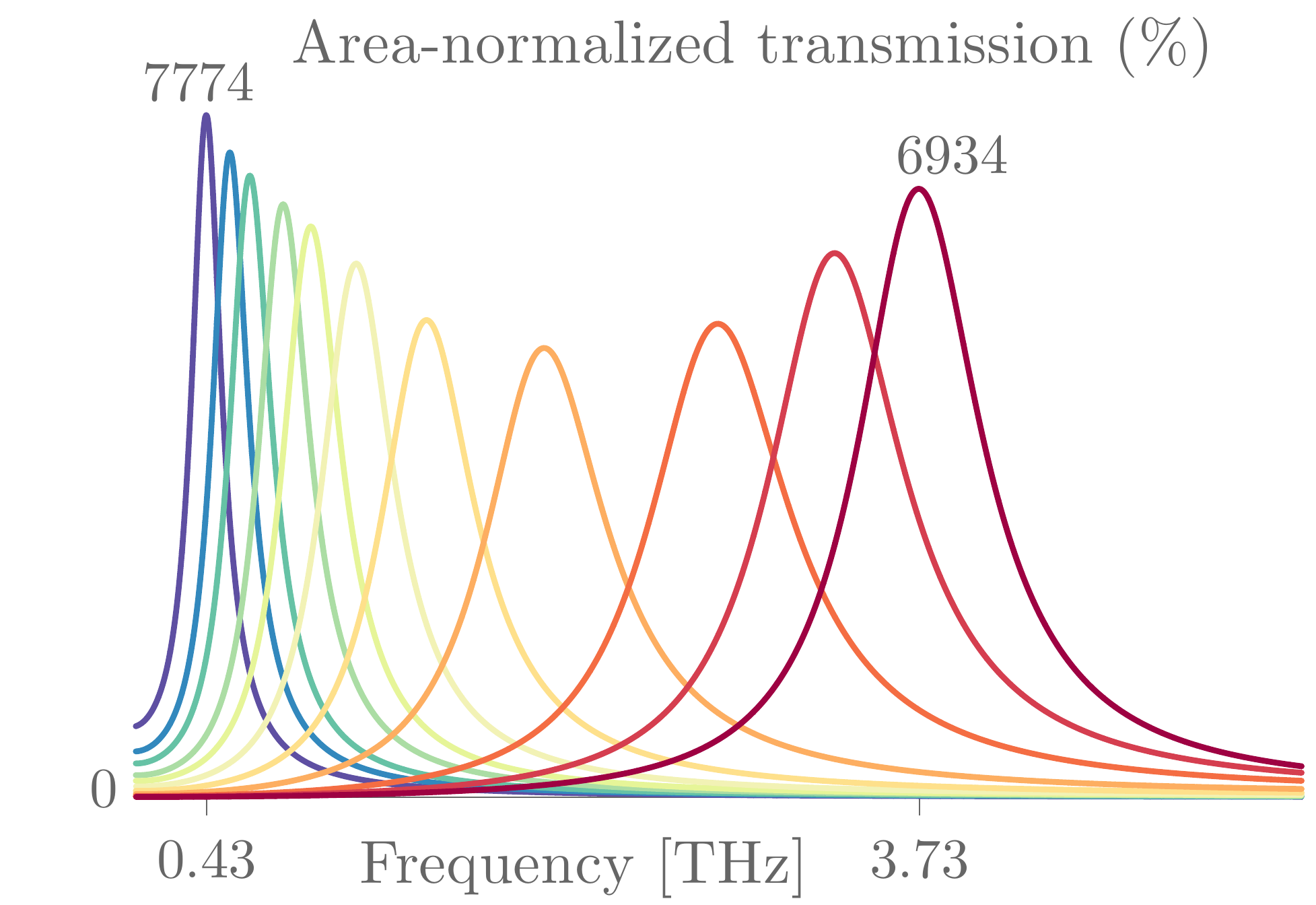}}
\caption{(a) Field enhancement (logscale) for perfect electric conductor (solid) and undamped Drude (dashed). (b) Field enhancement (logscale) for damped Drude. (c) Transmission (logscale) for damped Drude. (d) Area-normalized transmission for damped Drude. Legend is shared across all subfigures.}\label{fig:ringTheory}
\end{figure}


The simplest model assumes the gold film behaves as a perfect conductor with infinite conductivity. Prescribing perfect electric conductor ($\bE\times\bn = {\bf 0}$) conditions at the metal interface ensures the electric field is reflected at the metal boundary and no penetration is allowed. The field enhancement profile is presented in Fig. \ref{fig:ringTheory_PEC} with solid lines, exhibiting sharp peaks and enormous enhancements across gapsizes, showing that smaller gaps lead to larger resonances. This response corresponds to that of an undamped oscillator, which differs significantly to what has been observed experimentally for arrays of annular nanogaps \cite{baida2002light}. Quite interestingly, this unrealistic behavior may also be observed with the undamped Drude model ($\gamma = 0$). The field enhancement curves for this case, using $\hslash \overline{\omega}_p = 9.02$ eV  and $\varepsilon_\infty = 1$ adopted from Ordal \etal \cite{ordal1983optical,ordal1985optical}, are also depicted in Fig. \ref{fig:ringTheory_PEC} with dashed lines. We note that the maximum enhancement attained with undamped Drude and with PEC models is identical for a given gap size. Hence, the collision rate plays a pivotal role in the accurate characterization of the electromagnetic response through Drude's permittivity, since it is responsible for the imaginary component that models losses in the metals.

Secondly, we introduce damping in the Drude model with $\hslash \overline{\gamma} = 0.02678$ eV given by \cite{ordal1983optical,ordal1985optical}, otherwise known as the LRA. The losses introduced by a nonzero damping lead to lower field enhancements and broader resonances, see Fig. \ref{fig:ringFE_theory}, in comparison with both PEC and undamped Drude in Fig. \ref{fig:ringTheory_PEC}. Among distinct gap widths, these profiles are qualitatively similar, although smaller apertures lead to stronger field localizations and narrower resonance peaks. 

The metal is an opaque lossy medium, thus higher transmission rates are expected for wider gaps since light is only transmitted through the aperture in the metal, see Fig. \ref{fig:ringTP_theory}. In order to balance the extraordinary optical transmission among gapsizes,  transmission is normalized by the open area ratio $A_{\small w}/(A_{\small w} + A_{\scriptsize\mbox{Au}})$, see Fig. \ref{fig:ringATP_theory}. For instance, the annular 0.5 nm gap transmits a maximum of $0.13\%$ of incident light through an open area of $0.0016\%$, for an area-normalized transmission of $7774\%$, whereas the annular 100 nm gap is able to transmit $23\%$ of the incoming light through a wider open area of $0.33\%$, giving an area-normalized transmission of $6934\%$. Indeed, the normalized transmission for nanometric and sub-nanometric gaps is superior to that of nanogaps 100 times wider, as a consequence of the extreme amplification of the incident EM field that occurs for deep-subwavelength apertures.

\begin{figure}[h!]
\subcaptionbox{Relative blue-shift \label{fig:shiftTheory}}
[7.5cm]{\includegraphics[scale = 0.35]{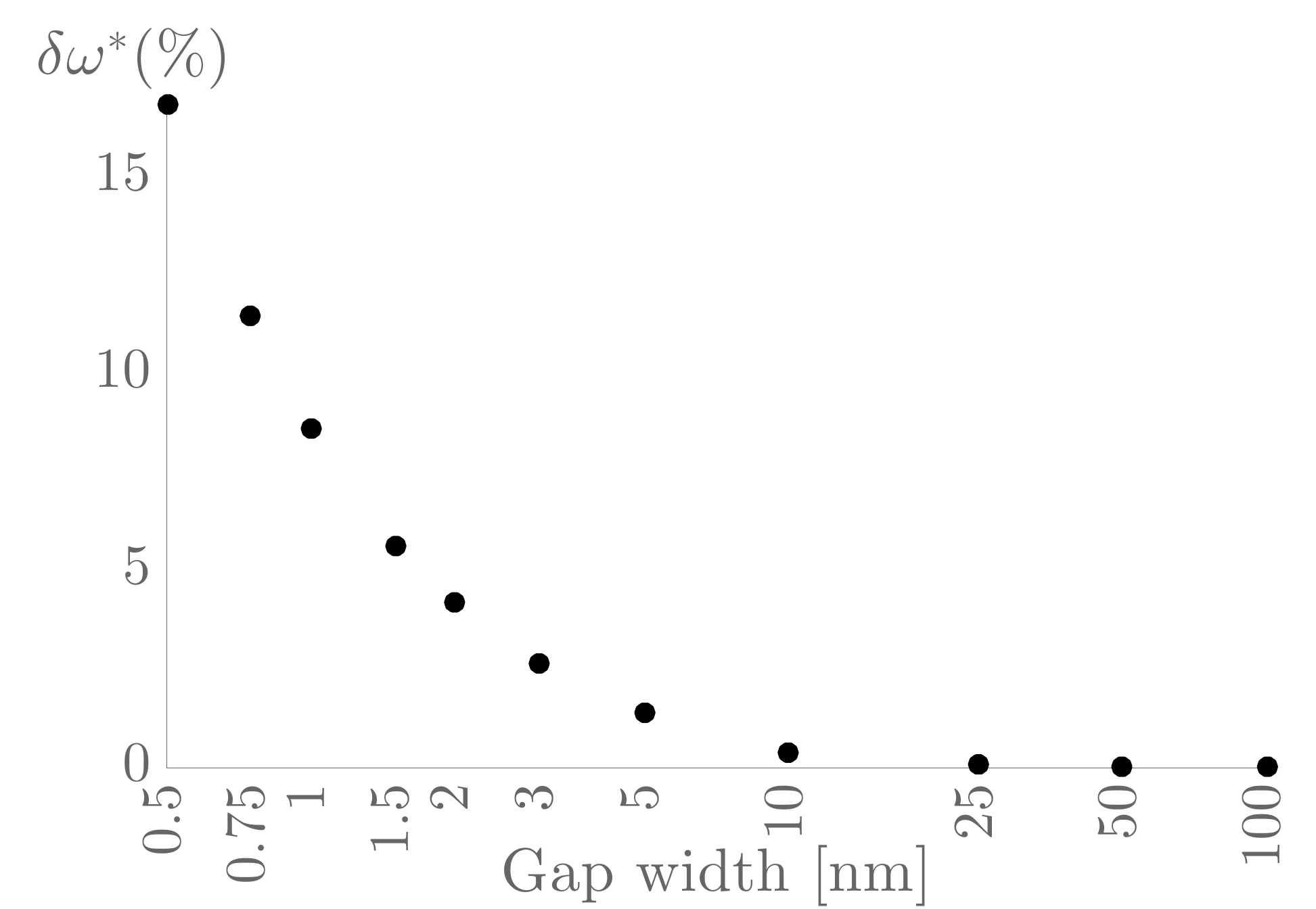}} 
\subcaptionbox{Ratios in field enhancement and transmission\label{fig:decayTheory}}
[7.5cm]{\includegraphics[scale = 0.35]{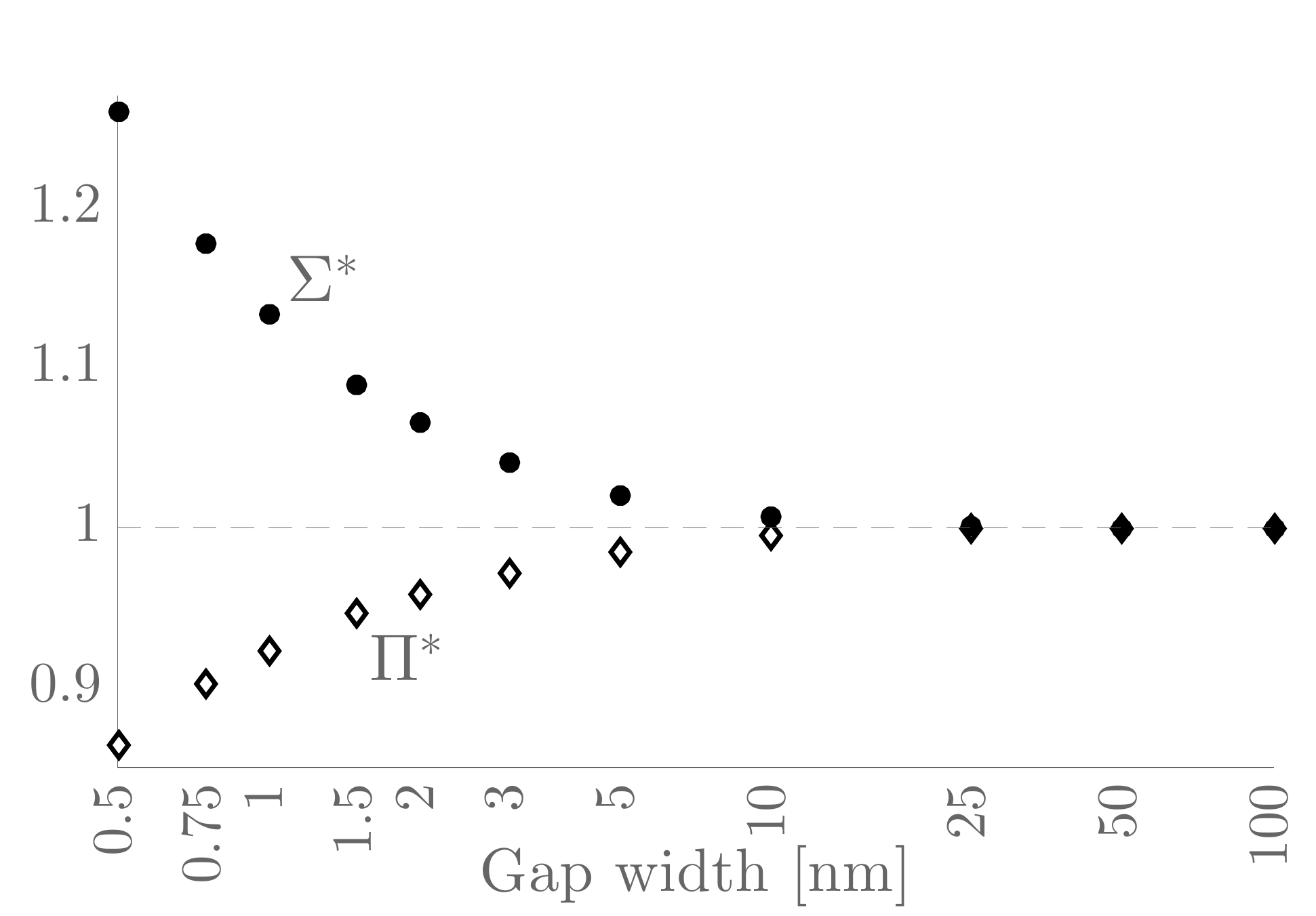}}\caption{Nonlocal effects are noticeable only at scales below 10 nm.}\label{fig:HMring}
\end{figure}

Finally, we extend the study above with the hydrodynamic model using $\widebar{v}_F = 1.39\cdot 10^6$ m/s. The nonlocal model for electron interaction leads to spectral changes that heavily depend on the gap width. The field enhancement and transmitted power profiles are qualitatively similar to those of the LRA in Figs. \ref{fig:ringFE_theory}-\ref{fig:ringATP_theory}, although quantitative discrepancies arise as we explore gaps below tenths of nanometers. To quantify the impact of the hydrodynamic model for the periodic annular nanogap, we evaluate the relative blue-shift $\delta\omega^* = (\omega^*_{HM}-\omega^*_{LRA})/\omega^*_{LRA}$ in the resonant frequency $\omega^*$, as well as the ratios of maximum field enhancement $\Pi^* = \fie^*_{HM}/\fie^*_{LRA}$  and maximum transmission $\Sigma^*= \pow^*_{HM}/\pow^*_{LRA}$, for multiple gap widths in Fig. \ref{fig:HMring}. Certainly, smaller gaps exhibit large shifts, even beyond 15\% for sub-nanometric widths, whereas the spectral response for gaps above 10 nm remains unchanged.

\begin{figure}[h!]
\centering
\includegraphics[scale = 0.35]{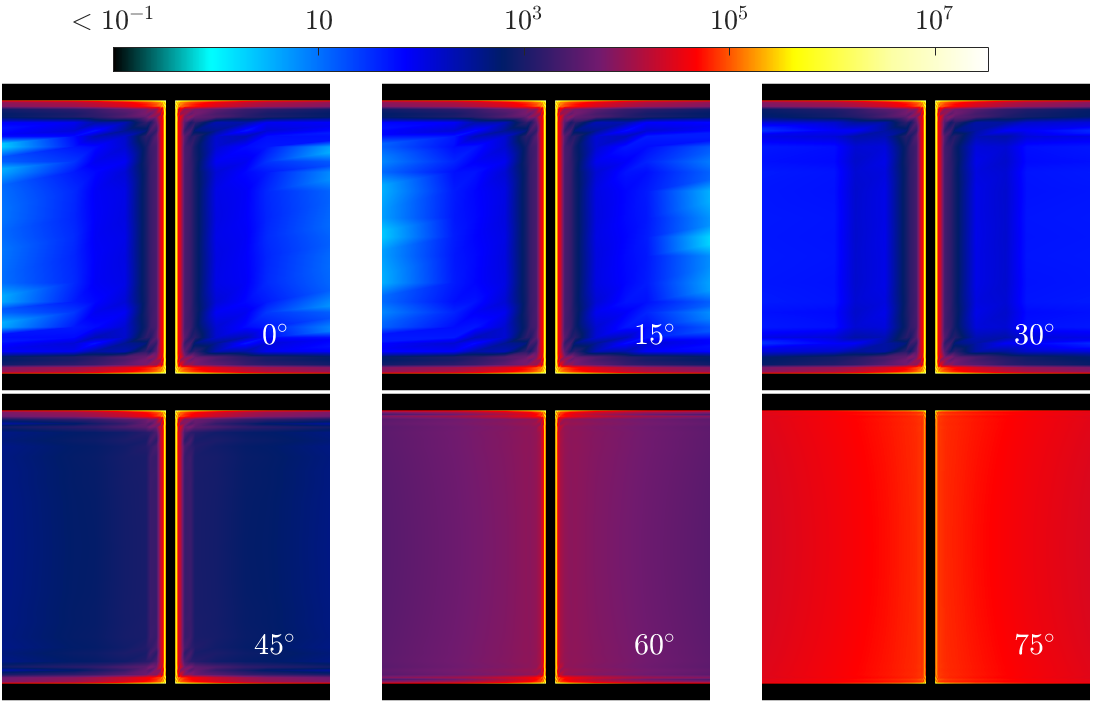}
\caption{Cross section view at several angles for 5 nm gap structure of solution field $\abs{\uprho_h}$ at the resonant frequency 1.45 THz, shown in logarithmic scale.}\label{fig:ringRho}
\end{figure}

These changes are a consequence of the spreading of the electron density at the metal interface explained above. For noble metals, such as gold, the smoothed profile of induced electron density causes an effective enlargement of the aperture seen by the incident EM wave. Larger effective gaps lead to resonance shifts towards the blue end of the spectrum, along with a decay in the maximum field enhancement (less confinement) and increment in maximum transmission (wider gap region). In Fig. \ref{fig:ringRho}, we inspect the induced charge density $\abs{\uprho}$ in the cross section of a 5 nm annular gap for several angular slices  $y/x = \tan \alpha$, specified in Fig. \ref{fig:annularFIR}. We observe for $\alpha = 0^\circ$ a boundary-layer pattern, with a maximum value at the interface and a decay of five orders of magnitude just a few nanometers away from the aperture. These two features gradually decrease as we move from the $y$-constant symmetry plane to the $x$-constant symmetry plane. Indeed, for $\alpha = 75^\circ$ the charge density profile is almost constant in the interior of the metal. Conversely, for the local model $\abs{\uprho}$ is infinitely localized at the gold surface, thus allowing less EM wave penetration in the metal.

These effects have been observed for nanoparticles and plasmonic dimers \cite{garcia2008nonlocal,raza2013blueshift,raza2015review,ciraci2012probing,ciraci2014film}, but have never been reported for neither annular structures nor at low THz frequencies. These results motivate the need to account for the hydrodynamic pressure in the simulation of realistic 3D plasmonic structures, since the nonlocal effects do have a substantial impact on the performance of the device for shrinking nanogaps.

\section{Conclusions} \label{sec:conc}
In this paper, we have presented a hybridizable discontinuous Galerkin method to simulate the propagation of electromagnetic waves for metal-dielectric media at the nanoscale. Simulation of plasmonic phenomena is inherently complex due to the enormous disparity in length scales and the extreme localization of electromagnetic fields that can be observed as a consequence of the collective excitation of electrons. The HDG method for Maxwell's equations, and the extension to the hydrodynamic model for metals are well-suited to the numerical simulation of plasmonic devices, due to its ability to handle complex geometries through anisotropic unstructured meshes, the efficient treatment of material interfaces and the possibility of solving reduced linear systems that only involve the degrees of freedom at the faces of the discretization. 


\section*{Acknowledgements}
F.~V.-C., N.~C.~N and J.~P. acknowledge support from the AFOSR Grant No. FA9550-11-1-0141 and the AFOSR Grant No. FA9550-12-0357. S.-H.~O. acknowledges support from the NSF Grant No. ECCS 1610333 and the Seagate Technology University Project. The authors thank Prof. Luis Mart\'{i}n-Moreno and Dr. Cristian Cirac\`{i} for their valuable suggestions, comments and inputs.


\end{document}